%% file: main.tex
\documentclass[letterpaper,11pt]{article}

\usepackage[margin=1in]{geometry}
\usepackage{setspace}
\usepackage[T1]{fontenc}
\usepackage[utf8]{inputenc}
\usepackage{lmodern}
\usepackage{microtype}
\usepackage[backend=biber,style=alphabetic,maxbibnames=100,maxalphanames=4]{biblatex}

\usepackage{amsmath,amssymb,amsthm,mathtools}
\usepackage{bm}
\usepackage{thmtools}
\usepackage{thm-restate}

\usepackage[dvipsnames,svgnames,x11names]{xcolor}
\usepackage[colorlinks]{hyperref}
\usepackage{enumerate}
\usepackage{nicefrac}
\usepackage{graphicx}
\usepackage{algorithm}
\usepackage{authblk}
\usepackage{esvect}
\usepackage[normalem]{ulem}
\usepackage{braket}
\usepackage{float}
\usepackage{pgfplots}
\usepackage{multicol}
\usepackage[nameinlink,noabbrev,capitalize]{cleveref}
\usepackage{adjustbox}
\usepackage{subcaption}
\usepackage[font=small]{caption}
\usepackage[toc,page]{appendix}
\usepackage{makecell}
\usepackage{array}
\usepackage{bm}
\usepackage{titlesec}

\titlespacing*{\paragraph}
  {0pt}   
  {0.1em} 
  {1em}   

\newcolumntype{C}[1]{>{\centering\arraybackslash}p{#1}}

\usetikzlibrary{calc}
\usetikzlibrary{trees}
\usetikzlibrary{patterns}
\usetikzlibrary{arrows.meta}

\pgfplotsset{compat=1.6}

\bibliography{bibliography}

\hypersetup{
    breaklinks = true,
	colorlinks = true,
	citecolor = teal,
	urlcolor = purple,
}

\newtheorem{theorem}{Theorem}

\usepackage{xcolor}

\Crefname{theorem}{Theorem}{Theorems}
\Crefname{theorem2}{Theorem}{Theorems}
\Crefname{lemma}{Lemma}{Lemmas}
\Crefname{lemma2}{Lemma}{Lemmas}
\Crefname{figure}{Fig.}{Figs.}
\Crefname{algorithm}{Listing}{Listings}
\Crefname{section}{Section}{Sections}
\Crefname{observation}{Observation}{Observations}
\Crefname{property}{Property}{Properties}
\Crefname{lemma}{Lemma}{Lemmas}
\Crefname{claim}{Claim}{Claims}
\Crefname{claimx}{Claim}{Claims}
\Crefname{figure}{Fig.}{Figs.}
\Crefname{subfigure}{Fig.}{Figs.} 
\Crefname{minipage}{Fig.}{Figs.}
\Crefname{enumi}{Condition}{Conditions}
\Crefname{reductionrule}{Rule}{Rules}

\usepackage[inline]{enumitem}

\definecolor{realblue}{rgb}{0,0,1}
\definecolor{defblue}{rgb}{0.274,0.392,0.666}
\definecolor{darkerblue}{rgb}{0.094,0.455,0.804}
\definecolor{darkblue}{rgb}{0.063,0.306,0.545}
\definecolor{linkblue}{rgb}{0.098,0.098,0.4392}
\definecolor{red}{rgb}{0.627,0.117,0.156}
\definecolor{green}{RGB}{51, 153, 102}
\definecolor{orange}{rgb}{0.903,0.739,0.382}
\definecolor{realred}{rgb}{1,0,0}
\definecolor{lipicsblue}{rgb}{0.08235294118,0.3098039216,0.537254902}

\usepackage{color,soul}

\usepackage{xspace} 
\usepackage[textsize=scriptsize]{todonotes}
\DeclareTextFontCommand{\emph}{\color{defblue}\em}
\DeclareTextFontCommand{\bl}{\color{lipicsblue}}
\hypersetup{colorlinks=true,
    linkcolor=lipicsblue,
    anchorcolor=lipicsblue,
    citecolor=lipicsblue,
    filecolor=lipicsblue,
    menucolor=lipicsblue,
    urlcolor=lipicsblue,
    bookmarksopen=true,
    bookmarksopenlevel=2,
    bookmarksnumbered=true,
    plainpages=false,
    }

\usepackage{tocloft}
\DeclareMathOperator{\Oh}{\mathcal O}

\DeclareMathOperator{\tw}{tw}
\DeclareMathOperator{\bw}{bw}
\DeclareMathOperator{\midset}{mid}

\DeclareMathOperator{\poly}{poly}

\newcommand{\jvchange}[1]{\textcolor{black}{#1}}

\title{Faster Planar Graph Algorithms\\ for Connectivity Problems via Meanders}

\author[1]{Susanna Caroppo}
\author[1]{Giordano Da Lozzo}
\author[1]{\authorcr Giuseppe Di Battista}
\author[2]{Jevgēnijs Vihrovs}

\affil[1]{Roma Tre University, Department of Engineering}
\affil[2]{University of Latvia, Faculty of Science and Technology}

\date{}

\begin{document}

\maketitle

\vspace{-8mm}
\begin{abstract} 
    In this paper, we refine the dynamic programming framework based on the sphere cut decomposition designed by Dorn, Penninkx, Bodlaender, and Fomin (ESA 2005) to obtain faster subexponential algorithms for connectivity problems on planar graphs.
    We investigate the relationship between these problems and meanders, which are simple closed planar loops that intersect a fixed line in a given number of points.
    By combining dynamic programming with techniques from meandric system analysis and the use of fast matrix multiplication by Dorn (ESA 2006), we obtain improved algorithms for planar connectivity problems.

    \jvchange{We show that the number of meanders on $2n$ crossings $M_n$ is $\Oh^*(12.806^n)$, which improves the previous upper bound of $\Oh^*(12.901^n)$ by Albert and Paterson (FPSAC 2004).}
    This then gives the best-known classical upper bounds on the deterministic time complexity of several planar graph problems with polynomially-bounded weights, namely $\Oh(2^{5.543\sqrt n})$ for the \textsc{Planar Travelling Salesman} problem, $\Oh(2^{5.796\sqrt n})$ for \textsc{Planar Longest Cycle/Path}, $\Oh(2^{8.251\sqrt n})$ for \textsc{Planar Connected Dominating Set} and $\Oh(2^{8.037\sqrt n})$ for \textsc{Planar Steiner Tree}.
    Notably, this leads to the best-known deterministic complexity $\Oh(2^{5.543\sqrt{n}})$ for the \textsc{Planar Hamiltonian Cycle} problem.
    
\end{abstract}

{
\setlength{\cftbeforesecskip}{0pt}
\scriptsize
\tableofcontents
}
\clearpage
\section{Introduction and Overview}\label{se:intro}

Planar graphs occupy a central position in graph theory and algorithms, both for their rich structural properties and for the algorithmic advantages they offer. In contrast to general graphs, many NP-hard problems admit significantly faster exact algorithms when restricted to $n$-vertex planar instances. Namely, the complexity of such algorithms can be subexponential, $2^{\Oh(\sqrt n \log n)}$ or even $2^{\Oh(\sqrt{n})}$.
A classic approach by Lipton and Tarjan to obtain such algorithms is to apply their famous {\em Planar Separator Theorem}~\cite{LT79} recursively and solve the problem using the divide-and-conquer technique \cite{LT80}.
Planar graphs, however, admit several algorithmic tools that go substantially beyond the separator paradigm. 
An influential technique for planar graphs and, more generally, for minor-free graphs is {\em Baker's layering method}~\cite{Bak94,Epp00}. By deleting one congruence class of BFS layers, a planar graph can be decomposed into components of bounded treewidth, which can then be optimized over by dynamic programming. This gives polynomial-time approximation schemes for a wide range of optimization problems on planar graphs, and has become one of the canonical examples of how planar structure can be algorithmically exploited~\cite{Epp00,ABFKN02,GB07,BKM09}.
Another major example is the {\em Bidimensionality Theory}~\cite{DFHT05}, which exploits the interaction between minor-monotonicity, treewidth, and grid-like structure. Together with the {\em Grid-Minor Theorem}~\cite{RST94} and its sharper forms for planar and minor-free graphs, bidimensionality yields powerful meta-theorems for fixed-parameter algorithms, subexponential algorithms, and polynomial-time approximation schemes on planar graphs and related families~\cite{DFHT04,DH05,DH08,FLST20}. 

An important class of planar graph problems is the set of \emph{connectivity problems} in which, broadly speaking, we are asked to find an optimal subgraph subject to topological constraints. 
This class of problems is the main focus of our work and includes relevant examples such as \textsc{Planar Hamiltonian Cycle}, \textsc{Planar Travelling Salesman}, and \textsc{Planar Steiner Tree}.
Whether the logarithmic factor in the exponent of the running time of exact separator-based algorithms for these problems could be removed was first answered positively for 
\textsc{Planar Hamiltonian Cycle} in \cite{DKW06}, where the complexity of the divide-and-conquer technique was brought down to $2^{\Oh(\sqrt n)}$, albeit with some large constant in the base of the exponent.
Since then, such complexity has been obtained for other planar connectivity problems as well \cite{DPBF10, CNPPVRW22, Pin16}, whose $2^{\Oh(\sqrt n)}$ complexity has also been proven to be optimal under the SETH conjecture \cite{CNPPVRW22}. 
The main open questions thus lie in determining the best constants for the $\Oh(\sqrt n)$ factor in the exponent. 

To reach this goal, modern techniques often combine dynamic programming (DP) with tools from the area of parameterized complexity.
For instance, dynamic programming over a tree decomposition of the input graph $G$ yields algorithms with running time $\Oh^*(2^{\Oh(\tw(G))})$ for a wide range of problems, where $\tw(G)$ is the treewidth of~$G$ \cite{Bod88}.
For planar graphs on $n$ vertices, $\tw(G) = \Oh(\sqrt{n})$, which directly leads to the aforementioned subexponential running times.

Dorn, Penninkx, Bodlaender, and Fomin showed a refinement of such dynamic programming using the \emph{sphere cut decomposition} of the input graph \cite{ST94}, leading to improved running times~\cite{DPBF10}.
Such a decomposition is defined for an embedding of the graph on the surface of a sphere.
Informally, the sphere cut decomposition can be viewed as a recursive binary partition of the surface into regions, where each region is bounded by a simple closed curve called a \emph{noose}.
They show that there exists such a decomposition in which every noose intersects at most $\bw(G)$ vertices, where $\bw(G)$ denotes the  \emph{branchwidth} of the graph, and it can be found efficiently \cite{GT08}.
Dynamic programming over this decomposition then yields algorithms with running time $\Oh^*(2^{\Oh(\bw(G))})$. Since it is known that $\bw(G) \leq \sqrt{4.5n} \leq 2.122\sqrt{n}$ \cite{FT06}, the desired subexponential complexity follows.

\paragraph{Current Bounds.} Among the results of \cite{DPBF10}, the authors obtained an $\Oh(2^{6.903\sqrt{n}})$-time algorithm for \textsc{Planar Hamiltonian Cycle} and \textsc{Planar Travelling Salesman}, an $\Oh(2^{7.223\sqrt{n}})$-time algorithm for \textsc{Planar Longest Path}, and an $\Oh(2^{9.264\sqrt{n}})$-time algorithm for \textsc{Planar Feedback Vertex Set}.
In a related work, Dorn showed that fast matrix multiplication results can be applied to speed up these dynamic programming algorithms; 
however, in the presence of weights, 
this incurs a multiplicative factor of $W$ in the complexity, where $W$ is the largest weight in the graph~\cite{Dor06}.
This implies, for example, that \textsc{Planar Travelling Salesman} can be solved in $\Oh(W\cdot 2^{5.579\sqrt{n}})$ time and \textsc{Planar Longest Path} in $\Oh(W \cdot 2^{5.828\sqrt{n}})$ time.
For unweighted problems, the factor $W$ disappears; therefore, for instance, \textsc{Planar Hamiltonian Cycle} can be solved in $\Oh(2^{5.579\sqrt{n}})$ time.

Subsequent works introduced more powerful techniques that considerably improved the complexity of planar connectivity algorithms, specifically \emph{Cut \& Count} \cite{CNPPVRW22} and the \emph{rank-based approach} \cite{BCKN15}.
Pino, Bodlaender and Van Rooij combined these techniques with branchwidth to obtain state-of-the-art complexities in the deterministic unbounded-weight setting and in the randomized polynomial-weight setting \cite{PBvR17}.
For example, they showed $\Oh(2^{6.570\sqrt{n}})$ deterministic time and $\Oh(W\cdot 2^{5.036\sqrt{n}})$ randomized time algorithms for the \textsc{Travelling Salesman} problem.
Nonetheless, the best deterministic algorithm in the polynomial-weight regime still remains Dorn's $\Oh(2^{5.571\sqrt{n}})$ algorithm \cite{Dor06} (when adjusted for the current best upper bound on the matrix multiplication constant $\omega < 2.372$ \cite{ADWXXZ25}).
For \textsc{Planar Hamiltonian Cycle}, until now, this has been the best-known deterministic complexity.
The authors of \cite{PBvR17} additionally applied their approach to the \textsc{Planar Connected Dominating Set}, \textsc{Planar Steiner Tree} and \textsc{Planar Feedback Vertex Set} problems.
Together with \cite{Dor06}, these two papers provide the state-of-the-art complexities for these problems, see \Cref{tab:sota}.

\paragraph{Meanders.}

A \emph{meander} of order $n$ is a Jordan curve that crosses a horizontal line in $2n$ points, see \Cref{fig:meander}.
Meanders are planar objects with various intriguing connections to different areas of mathematics, physics, and computer science, including computational geometry \cite{HMRT86}, knot theory \cite{BM20}, statistical physics \cite{DGG00}, as well as labyrinths and recreational mathematics \cite{Zvo24}.
The systematic study of meanders was initiated in \cite{LZ92,LZ93,Bac99}, and a lot of effort has gone into calculating the asymptotic growth $M_n$ of the number of meanders of order $n$ \cite{AP05}, called the \emph{meandric number}.
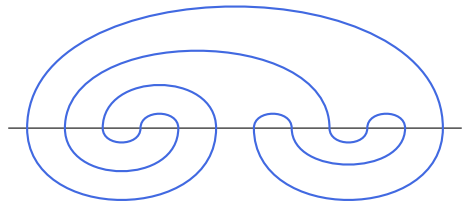
\begin{figure}[ht!]
    \centering
        \begin{tikzpicture}[scale=.5]
            \draw (-0.5,0) -- (11.5,0);
        
            \draw[thick,color=RoyalBlue]
            (0,0) to[out=90,in=90,looseness=1] (11,0)
                  to[out=-90,in=-90,looseness=1.3] (6,0)
                  to[out=90,in=90,looseness=1.3] (7,0)
                  to[out=-90,in=-90,looseness=1.1] (10,0)
                  to[out=90,in=90,looseness=1.3] (9,0)
                  to[out=-90,in=-90,looseness=1.3] (8,0)
                  to[out=90,in=90,looseness=1] (1,0)
                  to[out=-90,in=-90,looseness=1.3] (4,0)
                  to[out=90,in=90,looseness=1.3] (3,0)
                  to[out=-90,in=-90,looseness=1.3] (2,0)
                  to[out=90,in=90,looseness=1.3] (5,0)
                  to[out=-90,in=-90,looseness=1.3] (0,0);
        \end{tikzpicture}
    \caption{A meander of order 6.}
    \label{fig:meander}
\end{figure}

It has already been noticed in \cite{DPBF10} that the complexity of the dynamic programming algorithms for planar connectivity problems can be improved by ignoring some planar configurations that cannot lead to a solution.
Bobier and Sawada \cite{BS10} noticed that such configurations are closely related to meanders.
In this work, we explore this elegant connection in greater depth and show that planar connectivity problems admit improved algorithms by relating them to certain types of meanders and by excluding forbidden planar arrangements that cannot be related to such configurations.

In the algorithm of \cite{DPBF10} for the \textsc{Planar Hamiltonian Cycle}, the meanders implicitly appear as follows.
Conceptually, for each noose of the sphere cut decomposition, dynamic programming calculates all possible (topological) ways in which a partial Hamiltonian cycle of the corresponding subgraph (which is just a collection of paths that pass through every vertex in the subgraph) can intersect this noose.
The \emph{merging step} of the DP tries to combine two such partial solutions from the two adjacent nooses, which overlap in a continuous curve which we call the frontier.
The topology of these partial solutions can be such that combining them together induces a cycle disconnected from the rest of the graph.
Such a situation cannot result in a Hamiltonian cycle, hence it can be computationally ignored, reducing the running time of the algorithm.
The allowed configurations on the frontier are not meanders but rather \emph{meandric systems}, which allow for multiple disconnected curves, which are not necessarily closed.
Bobier and Sawada \cite{BS10} noticed that in the \textsc{Planar Hamiltonian Cycle} case, the relevant type of meandric systems are
\emph{acyclic meander systems}, i.e., those that avoid cycles (no closed curves are allowed). 
Denoting the asymptotic number of such systems on $n$ crossings by $\Theta^*(z^n)$, a trivial upper bound is $z \leq 4$.
The authors of \cite{DPBF10} gave a better bound of $z \leq 3.682$, which resulted in a faster algorithm for the \textsc{Planar Hamiltonian Cycle}.
Afterwards, Bobier and Sawada calculated the exact numbers of such systems up to $n = 28$, but did not prove a better upper bound on $z$.

\paragraph{Our Results.}
In this work, we use meanders and techniques from meander analysis to improve algorithms and their complexities for well-known planar graph connectivity problems, such as \textsc{Planar Travelling Salesman}, \textsc{Planar Longest Cycle/Path}, \textsc{Planar Connected Dominating Set}, \textsc{Planar Steiner Tree}, and \textsc{Planar Feedback Vertex Set}.

First, we show that the best-known upper bound on the meandric numbers from \cite{AP05} can be applied to the acyclic meandric systems.
By implementing the corresponding algorithm from \cite{AP05} and running the software on a modern machine, we obtain a refined upper bound on $z$.
In fact, our simulations provide the new best upper bound on the meandric numbers themselves, which improves upon their upper bound of $12.901$.
\begin{restatable}{theorem}{upperbound}\label{thm:upper-bound}
    The number of meanders of order $n$ is at most $\Oh^*(z^{2n}) = \Oh^*(12.806^n)$, and $z < 3.579$.
\end{restatable}

The new upper bound on $z$ improves the complexity upper bound of the algorithms for the \textsc{Planar Travelling Salesman} and \textsc{Planar Longest Cycle/Path} in~\cite{DPBF10}.
These two problems belong to what we call the \emph{cycle optimization} class of problems, in which we need to find an optimal cycle of the given graph.
Regarding the corresponding forbidden meandric systems, in such problems we are excluding configurations that contain multiple Jordan curves, see \Cref{fig:no-cycle-opt}.

\begin{figure}[ht!]
\centering

\begin{subfigure}[t]{0.3\textwidth}
    \centering
    \includegraphics[page=2,width=\textwidth]{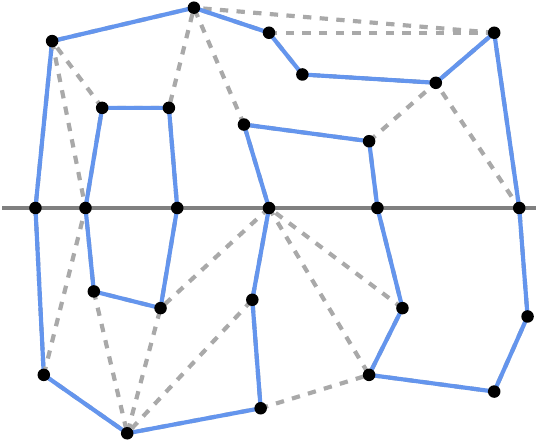}
    \caption{}
    \label{fig:no-cycle-opt}
\end{subfigure}
\hfill
\begin{subfigure}[t]{0.3\textwidth}
    \centering
    \includegraphics[page=3,width=\textwidth]{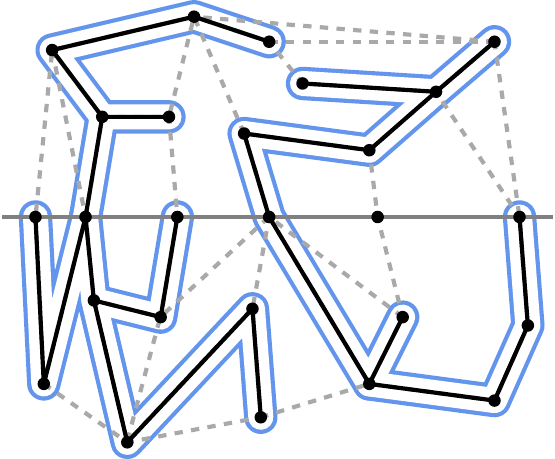}
    \caption{}
    \label{fig:no-tree-opt}
\end{subfigure}
\hfill
\begin{subfigure}[t]{0.3\textwidth}
    \centering
    \includegraphics[page=4,width=\textwidth]{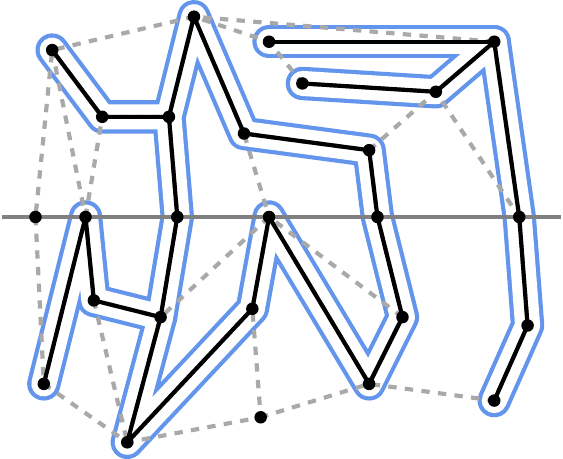}
    \caption{}
    \label{fig:no-forest-opt}
\end{subfigure}
\caption{ 
Illustration of some forbidden closed meandric configurations. The vertices of a graph are represented by small black circles and its edges are represented by solid and by dashed segments. The solid edges correspond to potential solutions, which induce forbidden situations.
(a) For the cycle optimization problems (e.g.\ \textsc{Planar Travelling Salesman}), the forbidden situation is a collection of cycles which corresponds to a closed meandric system with multiple loops;
(b) for the tree optimization problems (e.g.\ \textsc{Planar Steiner Tree}), a forbidden situation is, for instance, a subforest of the graph with multiple components which again corresponds to a closed meandric system with multiple loops; and 
(c) for the forest optimization problems (e.g.\ \textsc{Planar Feedback Vertex Set}), the forbidden situation is a subgraph containing a cycle which corresponds to a forbidden closed forest meandric system because it contains an enclosed loop.
}
\label{fig:no-instances}
\end{figure}

For planar graph problems like \textsc{Planar Connected Dominating Set} and \textsc{Planar Steiner Tree}, the goal is to find an optimal subtree of the given graph; we call the class of such problems the \emph{tree optimization} problems.
We adapt the meander technique to this case by interpreting the Eulerian tour of the subtree as a Jordan curve outline, which allows us to show improved DP algorithms for these problems as well.
In this case the forbidden meandric systems again are those that contain multiple Jordan curves, see \Cref{fig:no-tree-opt}.

Finally, for the \textsc{Planar Feedback Vertex Set} problem, it is known that it can be interpreted as finding some optimal subforest of the graph; problems such as this we place in the class of \emph{forest optimization} problems.
We then define and investigate the corresponding \emph{forest} meandric systems.
The forbidden instances in this case are configurations that contain a nested Jordan curve, see \Cref{fig:no-forest-opt}.
To count the number of such systems, we analyze the generating function of the corresponding formal language, similarly as in \cite{AP05}.
This gives an improved DP algorithm for this problem as well.

We then show that these methods can be combined with fast matrix multiplication \cite{ADWXXZ25} to obtain improved complexities in the setting with polynomially-weighted graphs in the same manner as in \cite{Dor06}.
All of our final classical complexities are given in \Cref{tab:results}.
The table describes upper bounds on the complexity of DP, augmented by either the fast matrix multiplication technique (FMM), or meanders, or both.
The FMM-based algorithms have a linear dependence on $W$, which is the maximum weight in the input graph.
We will refer to the setting in which $W = n^{\Oh(1)}$ as the \emph{bounded-weight} setting, and to the setting with arbitrarily-large weights as the \emph{unbounded-weight}~setting.

\begin{table}[ht!]
\begin{center}
\small
\begin{tabular}{p{5cm} || C{2.3cm} | C{2.3cm} || C{2.3cm} | C{2.3cm}}
     Problem & \makecell{DP\\ \cite{DPBF10}} & \makecell{DP\\+ Meanders\\(this paper)} & \makecell{DP + FMM\\\cite{Dor06}} & \makecell{DP + FMM\\ + Meanders\\(this paper)} \\ \hline\hline
     \textsc{Planar TSP}                & $\Oh(2^{6.987\sqrt n})$   & $\Oh(2^{6.876\sqrt n})$   & $\Oh(W 2^{5.655\sqrt n})$     & $\Oh(W 2^{5.543\sqrt n})$\\
     \textsc{Planar Longest Cycle/Path}                & $\Oh(2^{7.223\sqrt n})$   & $\Oh(2^{7.128\sqrt n})$   & $\Oh(W 2^{5.958\sqrt n})$     & $\Oh(W 2^{5.796\sqrt n})$\\
     \textsc{Planar Connected DS}       & $\Oh(2^{9.910\sqrt n})$   & $\Oh(2^{9.611\sqrt n})$   & $\Oh(W 2^{8.849\sqrt n})$     & $\Oh(W 2^{8.251\sqrt n})$\\
     \textsc{Planar Steiner Tree}       & $\Oh(2^{9.264\sqrt n})$   & $\Oh(2^{8.946\sqrt n})$   & $\Oh(W 2^{8.674\sqrt n})$     & $\Oh(W 2^{8.037\sqrt n})$\\
     \textsc{Planar Feedback VS}   & $\Oh(2^{9.264\sqrt n})$   & $\Oh(2^{9.119\sqrt n})$   & $\Oh(W 2^{8.674\sqrt n})$     & $\Oh(W 2^{8.382\sqrt n})$
\end{tabular}
\end{center}
\caption{Comparison between complexities obtained by combinations of different techniques: dynamic programming (DP, \cite{DPBF10}); dynamic programming together with fast matrix multiplication (DP + FMM, \cite{Dor06}); meandric systems (Meanders, this paper).
This table is dedicated to comparing the power of the techniques, as for some of the problems \cite{DPBF10,Dor06} have given slightly improved complexities using (weaker) meandric arguments.
The two leftmost columns are set in the unbounded-weight setting, and the two rightmost columns are set in the bounded-weight setting ($W$ denotes the maximum weight in the graph).
For the state-of-the-art complexities of the five problems, see \Cref{tab:sota}.}
\label{tab:results}
\end{table}

\textit{Bounded-Weight Setting.}
Our results improve the best-known deterministic complexity of several planar graph problems in the bounded-weight setting, including \textsc{Planar Travelling Salesman}, \textsc{Planar Logest Cycle/Path}, and \textsc{Planar Connected Dominating Set}, see \Cref{tab:sota}, where the improved complexities are underlined.
In particular, this gives the current best-known deterministic running time $\Oh(2^{5.543\sqrt{n}})$ for the \textsc{Planar Hamiltonian Cycle} problem.
However, we note that the randomized Monte Carlo algorithms of \cite{PBvR17} are still more efficient than our algorithms for all of the given problems.
Nonetheless, it is worth mentioning that since our complexities depend on the value of $z$, improving the upper bound on $z$ would immediately improve the asymptotic estimates.

We have also identified a number of discrepancies in the analysis of the running times in \cite{Dor06} for \textsc{Planar Connected Dominating Set}, \textsc{Planar Steiner Tree} and \textsc{Planar Feedback Vertex Set}, and a corner-case flaw in the algorithm for \textsc{Planar Steiner Tree} in \cite{DPBF10}.
Those have resulted in incorrect smaller complexities; 
our tables contain the actual complexities after making necessary corrections.
We describe the issues in detail in \Cref{app:mistakes}, where we also provide corrections in the complexity analysis.

\begin{table}[bt!]
\begin{center}
\small
\begin{tabular}{p{5cm} || C{2.3cm} | C{2.3cm} | C{2.3cm} || C{2.3cm}}
     Problem & \makecell{Deterministic\\\cite{Dor06}} & \makecell{Deterministic\\(this paper)} & \makecell{Randomized\\ \cite{PBvR17}} & \makecell{Deterministic\\ \cite{PBvR17}} \\ \hline\hline
     \textsc{Planar TSP}                & $\Oh(W2^{5.571\sqrt n})$           & $\underline{\Oh(W2^{5.543\sqrt n})}$   & $\Oh(W2^{5.034\sqrt n})$ & $\Oh(2^{6.570\sqrt n})$\\
     \textsc{Planar Longest Cycle/Path}                & $\Oh(W 2^{5.819\sqrt n})$   & \underline{$\Oh(W 2^{5.796\sqrt n})$}   &  $~\Oh(W2^{5.034\sqrt n})^\ddag$ & ~$\Oh(2^{6.912\sqrt n})^{\ddag}$\\
     \textsc{Planar Connected DS}       & ~$\Oh(W2^{8.674\sqrt n})^{\dagger}$   & $\underline{\Oh(W2^{8.251\sqrt n})}$   & $\Oh(W2^{5.034\sqrt n})$ & $\Oh(2^{8.777\sqrt n})$\\
     \textsc{Planar Steiner Tree}       & ~$\Oh(W2^{8.674\sqrt n})^{\dagger}$   & $\underline{\Oh(W2^{8.037\sqrt n})}$   & $\Oh(W2^{3.989\sqrt n})$ & $\Oh(2^{8.038\sqrt n})$\\
     \textsc{Planar Feedback VS}   & ~$\Oh(W2^{8.674\sqrt n})^{\dagger}$   & $\Oh(W2^{8.382\sqrt n})$   & $\Oh(W2^{3.989\sqrt n})$ & $\Oh(2^{8.038\sqrt n})$
\end{tabular}
\end{center}
\caption{Comparison of our results with the previous state-of-the-art complexities. 
The symbol $\dagger$ denotes corrected complexities (see \Cref{app:mistakes}), fixing discrepancies in the analysis in previous works.
The underlined entries are results of this paper that improve the state-of-the-art complexities for the corresponding problems, 
except for the last, where the existing unbounded-weight algorithm is more efficient.
The bounds from the previous results were updated using the best-known value $\omega < 2.372$ \cite{ADWXXZ25}.
The complexity denoted by $^\ddag$ follows from \cite{PBvR17}, see \Cref{app:longest}.}
\label{tab:sota}
\end{table}

\section{Preliminaries}

In this section, we provide preliminaries and basic definitions. For standard concepts and terminology about graphs and their embeddings, we refer the reader to~\cite{Die12}. Unless stated otherwise, all graphs in this paper are connected, finite, and simple.

\paragraph{Graph Notation.}
For an undirected graph $G=(V,E)$, we denote by $V(G)$ and $E(G)$ the set $V$ of vertices and the set $E$ of edges of $G$, respectively, and we denote by $n$ the cardinality of $V(G)$. We also denote by $\delta(G)$ the degree of the vertex of $V(G)$ with the minimum degree. 
A curve is \emph{simple} if it does not self-intersect.
A \emph{Jordan arc} (resp.\ \emph{Jordan curve}) is an open (resp.\ closed) simple curve in the plane.
A \emph{planar drawing} of a graph $G$ in the plane 
maps each vertex $v$ of $G$ to a point $p(v)$ and each edge $\{u,v\} \in E(G)$ to a Jordan arc connecting $p(u)$ and $p(v)$ 
that contains no other vertex and does not intersect any other arc, except at a common endpoint.
A graph is \emph{planar} if it admits a planar drawing.
Analogously, we can define a \emph{spherical drawing} as a drawing of $G$ on the surface of a sphere. It is easy to see that a graph is planar if and only it admits a spherical drawing.
An edge $e\in E(G)$ is a \emph{bridge} if its removal increases the number of connected components of $G$. A graph is said to be \emph{bridgeless} if it contains no bridges.

\paragraph{Problems.}
In this paper, we consider the following five problems.
For a planar graph $G=(V,E)$, let $w:E\rightarrow \mathbb{N}$ be a weight function defined on the edges of~$G$.
In the {\sc Planar Travelling Salesman} problem (Planar TSP), we seek the minimum length $\sum_{e\in C}w(e)$ of a cycle $C$ in $G$ that passes through all vertices in $V$ exactly once (a Hamiltonian Cycle). 
In the {\sc Planar Longest Cycle/Path} problem,
we seek the maximum weight of a cycle/path $C$ in $G$. In the {\sc Planar Steiner Tree} problem, we are given a subset $X\subseteq V$ of \emph{terminal} vertices and we seek the minimum total edge weight of a connected subgraph $H\subseteq G$ such that $X \subseteq V(H)$. 

For other problems, let $w:V\rightarrow \mathbb{N}$ be a vertex weight function.
In the {\sc Planar Connected Dominating Set} problem, we seek the minimum total weight of a set $D\subseteq V$ such that for any vertex $v \in V$, either $v \in D$ or $v$ is \emph{dominated} by a neighbor, i.e., a neighbor of $v$ is in $D$, and the subgraph induced by $D$ is connected.
Finally, in the {\sc Planar Feedback Vertex Set} problem, we seek the minimum weight $\sum_{v\in Y}w(v)$ of a set $Y\subseteq V$ such that every cycle in $G$ contains at least one vertex from $Y$.

In \cref{app:problems}, we show how all these problems can be reduced in polynomial time to the case of bridgeless graphs.

\paragraph{Sphere Cut Decomposition.}

Let $G$ be a planar graph on $n$ vertices.
A \emph{branch decomposition} of $G$ is a pair $(T,\mu)$, where $T$ is an unrooted ternary tree and $\mu: L \leftrightarrow E(G)$ is a bijection between the set of leaves $L$ of $T$ and the edges of $G$.
For an edge $e \in E(T)$, define $T_1$ and $T_2$ as the disconnected subtrees obtained by removing $e$ from $T$.
Let $G_i$ be the subgraph of $G$ induced by the edge set $\{\mu(f) \mid f \in L \cap V(T_i)\}$ for $i \in \{1,2\}$.
The \emph{middle set} of an edge $e \in T$ is then defined as $\midset(e) = V(G_1) \cap V(G_2)$.
The \emph{width} of the branch decomposition $(T,\mu)$ is the maximum of $\left|\midset(e)\right|$ over all $e \in T$.
The \emph{branchwidth} of $G$ is the minimum width over all branch decompositions of $G$, denoted by $\bw(G)$ (we will sometimes shorthand it simply by $\bw$, as the argument will always~be~$G$).
It is well-known that for an $n$-vertex planar graph $G$, we have $\bw(G)\leq \sqrt{4.5 n} \leq 2.122\sqrt n$~\cite{FT06}.

Consider a \emph{topological graph} $G$, i.e., a planar graph with an associated spherical drawing.  
The connected components of the sphere obtained by removing $G$ from it are called \emph{faces}.
An \emph{$O$-arc} is a Jordan curve on the sphere.
A \emph{noose} of $G$ is an $O$-arc that meets $G$ only at vertices and such that its intersection with a face is either empty or connected.
A \emph{sphere cut decomposition} of $G$ is a branch decomposition $(T,\mu)$ equipped with an ordering $\pi$, where for each $e \in T$ there exists a noose $O_e$ such that $O_e = \midset(e)$ and it bounds two open disks $\Delta_1$ and $\Delta_2$ of the sphere so that $G_i$ lies entirely in $O_e \cup \Delta_i$.
For each $e \in T$, $\pi(e)$ is a cyclic ordering of the vertices of $\midset(e)$ on the noose $O_e$ (clockwise or counter-clockwise).
See \Cref{fig:noose} for an illustration.

\begin{figure}[ht!]
    \centering
    \includegraphics[scale=0.6]{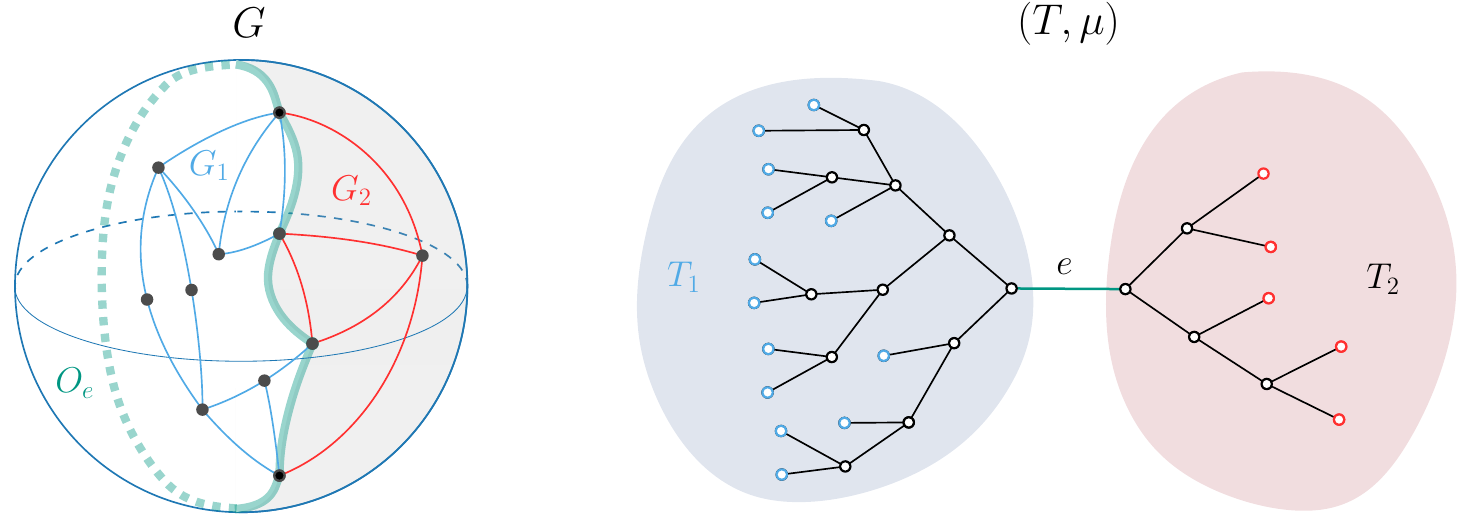}
    \caption{A topological graph $G$ embedded on a sphere and a sphere cut decomposition of $G$.}
    \label{fig:noose}
\end{figure}

An optimal sphere cut decomposition can be computed efficiently:

\begin{theorem}[\cite{ST94,GT08,DPBF10,MP15}] \label{thm:sc-decomposition}
Let $G$ be an $n$-vertex, connected, bridgeless, topological graph. There exists a sphere cut decomposition of $G$ of width $\bw(G)$ and such a decomposition can be constructed in $\Oh(n^3)$ time.
\end{theorem}

\paragraph{Non-Crossing Partitions and Matchings.}
Let $S = \{1,\dots,n\}$ and let $P(n) = \{P_1,\dots,P_m\}$ be a partition of $S$.
Consider a geometric representation of $P(n)$ 
such that the elements of $S$ are represented by $n$ points labeled $1,\dots,n$ in clockwise order along a circle,
and each set $P_i$ is represented by the convex polygon (possibly consisting of a single point or of a single edge) whose vertices correspond to the elements of $P_i$. The partition $P(n)$ is \emph{non-crossing} if the polygons representing the sets do not intersect. It is easy to see that $P(n)$ is non-crossing if there do not exist $a < b < c < d$ such that $a,c \in P_i$ and $b,d \in P_j$ for some $i \neq j$. 
A partition in which every set $P_i$ has size two is a \emph{matching}. Referring to the above geometric representation, each matched pair is connected by a straight-line segment.
It is well-known that the numbers of non-crossing partitions and matchings on $n$ elements are the $n$-th and $n/2$-th Catalan numbers, respectively~\cite{Kre72}.

Non-crossing partitions and non-crossing matchings can be uniquely represented by balanced bracket sequences in the following way \cite{HMRT86}.
Lay out the circle and the $n$ points on a horizontal line segment.
For a matching, draw each edge as an arc connecting the corresponding points above this segment.
For a partition, each point is assigned a pair of brackets, one directly to the left and one to the right.
For a set $P_i$ of the partition, the first point is assigned a pair of brackets $(\texttt{[}, \texttt{[})$, the last point $(\texttt{]}, \texttt{]})$, and all points in the middle are assigned $(\texttt{]}, \texttt{[})$.
The sides of the polygon corresponding to $P_i$ then are drawn as arcs between the corresponding brackets, see \Cref{fig:brackets}.
The only exception is when $P_i$ is a singleton, in which case the corresponding point is assigned a pair $(\texttt{[},\texttt{]})$.
The resulting bracket sequence in both cases is balanced.
The place where the circle is cut defines the starting point of the bracket labeling.
In our applications, $P_i$ will often represent a connected subgraph in the interior of the circle that touches the circumference only at the points of $P_i$.
Thus, the shaded regions of the figure represent the topology of the connectivity of the subgraphs.

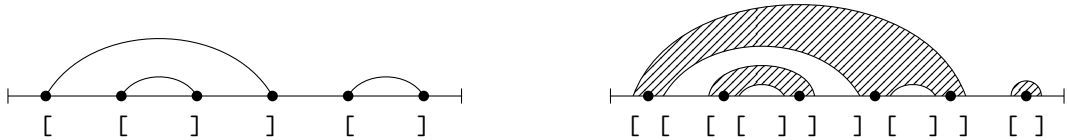
\begin{figure}[ht!]
    \centering
    \begin{subfigure}[t]{0.475\textwidth}
        \centering
        \begin{tikzpicture}
            \draw (-0.5,0) -- (5.5,0);
            \draw (-0.5,-.1) -- (-0.5,+.1);
            \draw (+5.5,-.1) -- (+5.5,+.1);
        
            \foreach \x in {0,1,2,3,4,5}
                \fill (\x,0) circle (2pt);
        
            \draw (0,0) to[out=60,in=120] (3,0);
            \draw (1,0) to[out=60,in=120] (2,0);
            \draw (4,0) to[out=60,in=120] (5,0);

            \small
            
            \node[below=4pt] at (0,0) {\texttt{[}};
            \node[below=4pt] at (1,0) {\texttt{[}};
            \node[below=4pt] at (2,0) {\texttt{]}};
            \node[below=4pt] at (3,0) {\texttt{]}};
            \node[below=4pt] at (4,0) {\texttt{[}};
            \node[below=4pt] at (5,0) {\texttt{]}};
        \end{tikzpicture}
    \end{subfigure}
    \begin{subfigure}[t]{0.475\textwidth}
        \centering
        \begin{tikzpicture}[scale=1]
            \draw (-0.5,0) -- (5.5,0);
            \draw (-0.5,-.1) -- (-0.5,+.1);
            \draw (+5.5,-.1) -- (+5.5,+.1);
        
            \foreach \x in {0,1,2,3,4,5}
                \fill (\x,0) circle (2pt);
        
            \def\eps{0.2}

            \coordinate (L1) at (-\eps,0);
            \coordinate (R1) at (\eps,0);
            \coordinate (L2) at (1-\eps,0);
            \coordinate (R2) at (1+\eps,0);
            \coordinate (L3) at (2-\eps,0);
            \coordinate (R3) at (2+\eps,0);
            \coordinate (L4) at (3-\eps,0);
            \coordinate (R4) at (3+\eps,0);
            \coordinate (L5) at (4-\eps,0);
            \coordinate (R5) at (4+\eps,0);
            \coordinate (L6) at (5-\eps,0);
            \coordinate (R6) at (5+\eps,0);
            
            \draw (L1) to[out=70,in=110] (R5);
            \draw (R1) to[out=60,in=120] (L4);
            \draw (L2) to[out=80,in=100] (R3);
            \draw (R2) to[out=60,in=120] (L3);
            \draw (R4) to[out=60,in=120] (L5);
            \draw (L6) arc[start angle=180,end angle=0,radius=\eps] (R6);
        
            \fill[pattern=north east lines, even odd rule]
                (L1) to[out=70,in=110] (R5) -- cycle
                (R1) to[out=60,in=120] (L4) -- cycle
                (R4) to[out=60,in=120] (L5) -- cycle;
                
            \fill[pattern=north east lines, even odd rule]
                (L2) to[out=80,in=100] (R3) -- cycle
                (R2) to[out=60,in=120] (L3) -- cycle;
                
            \fill[pattern=north east lines, even odd rule]
                (L6) arc[start angle=180,end angle=0,radius=\eps] (R6) -- cycle;

            \small
            
            \node[below=4pt] at (L1) {\texttt{[}};
            \node[below=4pt] at (R1)  {\texttt{[}};
        
            \node[below=4pt] at (L2) {\texttt{[}};
            \node[below=4pt] at (R2) {\texttt{[}};
        
            \node[below=4pt] at (L3) {\texttt{]}};
            \node[below=4pt] at (R3) {\texttt{]}};
        
            \node[below=4pt] at (L4) {\texttt{]}};
            \node[below=4pt] at (R4) {\texttt{[}};
        
            \node[below=4pt] at (L5) {\texttt{]}};
            \node[below=4pt] at (R5) {\texttt{]}};
        
            \node[below=4pt] at (L6) {\texttt{[}};
            \node[below=4pt] at (R6) {\texttt{]}};
        \end{tikzpicture}
    \end{subfigure}
    
    \caption{The bracket labelings of a non-crossing matching (left) and a non-crossing partition (right).}
    \label{fig:brackets}
\end{figure}

\paragraph{Meanders.}
A \emph{meander} of order $n$ is a Jordan curve that crosses a horizontal line in $2n$ points, see \Cref{fig:systems}.
Two meanders are equivalent if one can be obtained from the other by stretching and moving the curves without creating new intersections.
The number of meanders of order $n$ is called the \emph{meandric number} and is denoted by $M_n$.
An explicit formula for the meandric numbers is not known; nonetheless, close lower and upper bounds are known.
Albert and Paterson showed that the limit $M := \lim_{n \to \infty} M_n^{1/n}$ exists and $11.380 \leq M \leq 12.901$ \cite{AP05}.
Using numeric calculations, Jensen and Guttmann gave an estimate $M \approx 12.26$ \cite{JG00}.

A \emph{meandric system} with $n$ crossings is a collection of pairwise non-intersecting Jordan arcs and Jordan curves that cross a horizontal line in $n$ points, see again \Cref{fig:systems}. Additionally, the open ends of the Jordan arcs must extend either to the left or the right direction; if the ends are allowed to extend in both directions, we call a system \emph{bidirectional}; if the ends are allowed to extend only in one direction, we call a system \emph{unidirectional}, and w.l.o.g.~assume that the ends are extended to the right.
We will use the terminology of \cite{BS10} with a slight deviation: we will call a meandric system \emph{open} if Jordan arcs are allowed, and \emph{closed} otherwise; \emph{acyclic} if it contains no Jordan curves (the \cite{BS10} definition of open systems includes the acyclicity).

\begin{figure}[!ht]
    \begin{subfigure}[t]{0.475\textwidth}
        \centering
        \tiny
        \begin{tikzpicture}[scale=.7]
            \draw (-0.5,0) -- (7.5,0);
        
            \foreach \x in {0,1,2,3,4,5,6,7}
                \fill (\x,0) circle (2pt);
        
            \draw
            (0,0) to[out=90,in=90,looseness=1.2] (3,0)
                  to[out=-90,in=-90,looseness=1.3] (2,0)
                  to[out=90,in=90,looseness=1.3] (1,0)
                  to[out=-90,in=-90,looseness=1.1] (4,0)
                  to[out=90,in=90,looseness=1.3] (5,0)
                  to[out=-90,in=-90,looseness=1.3] (6,0)
                  to[out=90,in=90,looseness=1.3] (7,0)
                  to[out=-100,in=-80,looseness=0.9] (0,0);

            \node[left, yshift=+0.2cm] at (0,0) {\texttt{[}};
            \node[left, yshift=+0.2cm] at (1,0) {\texttt{[}};
            \node[right, yshift=+0.2cm] at (2,0) {\texttt{]}};
            \node[right, yshift=+0.2cm] at (3,0) {\texttt{]}};
            \node[left, yshift=+0.2cm] at (4,0) {\texttt{[}};
            \node[right, yshift=+0.2cm] at (5,0) {\texttt{]}};
            \node[left, yshift=+0.2cm] at (6,0) {\texttt{[}};
            \node[right, yshift=+0.2cm] at (7,0) {\texttt{]}};

            \node[left, yshift=-0.2cm] at (0,0) {\texttt{[}};
            \node[left, yshift=-0.2cm] at (1,0) {\texttt{[}};
            \node[left, yshift=-0.2cm] at (2,0) {\texttt{[}};
            \node[right, yshift=-0.2cm] at (3,0) {\texttt{]}};
            \node[left, yshift=-0.2cm] at (4,0) {\texttt{]}};
            \node[left, yshift=-0.2cm] at (5,0) {\texttt{[}};
            \node[right, yshift=-0.2cm] at (6,0) {\texttt{]}};
            \node[right, yshift=-0.2cm] at (7,0) {\texttt{]}};
        \end{tikzpicture}
    \end{subfigure}
    \quad
    \begin{subfigure}[t]{0.475\textwidth}
        \centering
        \tiny
        \begin{tikzpicture}[scale=.7]
            \draw (-0.5,0) -- (7.5,0);
        
            \foreach \x in {0,1,2,3,4,5,6,7}
                \fill (\x,0) circle (2pt);
        
            \draw
            (-0.5,+1) to[out=0,in=90,looseness=1] (2,0)
                      to[out=-90,in=-90,looseness=1.2] (5,0)
                      to[out=90,in=90,looseness=1.2] (6,0)
                      to[out=-90,in=180,looseness=1.2] (+7.5,-1);
            \draw
            (+7.5,-1.5) to[out=180,in=-90,looseness=.8] (1,0)
                      to[out=90,in=90,looseness=1.2] (0,0)
                      to[out=-90,in=180,looseness=.8] (+7.5,-2);
            \draw
            (3,0)     to[out=90,in=90,looseness=1.2] (4,0)
                      to[out=-90,in=-90,looseness=1.2] (3,0);
            \draw
            (+7.5,+.5) to[out=180,in=90,looseness=1.2] (7,0)
                      to[out=-90,in=180,looseness=1.2] (+7.5,-.5);

            \node[left, yshift=+0.2cm] at (0,0) {\texttt{[}};
            \node[right, yshift=+0.2cm] at (1,0) {\texttt{]}};
            \node[right, yshift=+0.2cm] at (2,0) {\texttt{]}};
            \node[left, yshift=+0.2cm] at (3,0) {\texttt{[}};
            \node[right, yshift=+0.2cm] at (4,0) {\texttt{]}};
            \node[left, yshift=+0.2cm] at (5,0) {\texttt{[}};
            \node[right, yshift=+0.2cm] at (6,0) {\texttt{]}};
            \node[left, yshift=+0.2cm] at (7,0) {\texttt{[}};

            \node[left, yshift=-0.2cm] at (0,0) {\texttt{[}};
            \node[left, yshift=-0.2cm] at (1,0) {\texttt{[}};
            \node[left, yshift=-0.2cm] at (2,0) {\texttt{[}};
            \node[left, yshift=-0.2cm] at (3,0) {\texttt{[}};
            \node[right, yshift=-0.2cm] at (4,0) {\texttt{]}};
            \node[right, yshift=-0.2cm] at (5,0) {\texttt{]}};
            \node[left, yshift=-0.2cm] at (6,0) {\texttt{[}};
            \node[left, yshift=-0.2cm] at (7,0) {\texttt{[}};
        \end{tikzpicture}
    \end{subfigure}
    
    \caption{A meander of order $4$ (left) and a bidirectional open meandric system that is not acyclic (right), with the respective bracket labelings.}
    \label{fig:systems}
\end{figure}
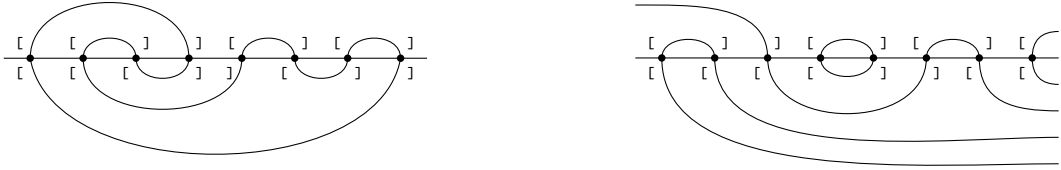

We can use the bracket labeling of non-crossing matchings also for the meandric systems; each crossing point has one bracket on each half-plane defined by the line.
Consider the half-plane above the line; the one below is analogous. An opening bracket \texttt{[} (resp.\ closing bracket \texttt{]}) besides the point denotes that either the next intersection of the arc with the line is to the right (resp.\ left) or that the arc extends to the right (resp. left).
Note that for an arbitrary meandric system the labelings on each half-plane will not necessarily be balanced bracket sequences; unmatched brackets correspond to the open ends of the arcs extending to the left and to the right.

\paragraph{Complexity Notation.}

To denote the complexity of our algorithms, we use the standard big-O notation, which we denote by $\Oh$.
We write $\Oh^*(f(n))$ to denote $\Oh(f(n) \poly(n))$ and $\widetilde\Oh(f(n))$ to denote $\Oh(f(n) \poly \log(n))$.
We study the classical algorithms in the standard word RAM model, in which if the size of the input is $N$, the arithmetic operations on $\Oh(\log N)$-bit integers and memory operations are assumed to take unit time.


\section{Faster Dynamic Programming over Sphere Cut Decomposition}

In this section, we first describe the dynamic programming scheme of Dorn, Penninkx, Bodlaender and Fomin \cite{DPBF10} on the sphere cut decomposition and its enhancement using fast matrix multiplication, as shown by Dorn~\cite{Dor06}.
Our review of these tools and of their analysis is necessary to describe our approach to speed up both these frameworks using meandric systems. 
A reader familiar with the frameworks of \cite{DPBF10,Dor06} may continue to \Cref{sec:speedup}, after getting familiar with our notation introduced in the first part of this section. 
To facilitate the understanding of the general scheme and of our meandric framework, in this section, we describe them for solving the \textsc{Planar Travelling Salesman} problem.
This approach can be applied to other cycle optimization problems (a class of planar graph problems in which we need to find an optimal cycle in the graph), for example, we adapt it to \textsc{Planar Longest Cycle/Path} at the end of the section.
In the subsequent sections, we will show how to employ the framework to design improved exact algorithms for several other planar graph problems.

\subsection{The Dynamic Programming Algorithm} \label{sec:dp}

In this section, we describe the dynamic programming algorithm of~\cite{DPBF10}, with slight modifications for our purposes.
The algorithm first computes an optimal sphere cut decomposition $(T,\mu,\pi)$, see \Cref{thm:sc-decomposition}.
Let $\ell = \bw(G)$ be the width of this decomposition.
The algorithm then will iterate over all edges $r \in G$ and compute the length of the optimal Hamiltonian cycle that passes through $r$ if such exists; this will add only $\Oh(n)$ multiplicative factor in the complexity.
Let $e_{\rho} = \{\rho_1,\rho_2\} \in T$ be the edge such that $\rho_1$ is a leaf in $T$ and $\mu(\rho_1) = r$.
The noose $O_{\rho}$ that corresponds to this edge partitions the sphere into two disks, where one of those disks contains only the edge $r$, and the other contains $G \setminus r$.
We call $O_{\rho}$ the \emph{root noose}.

Root the tree $T$ at the vertex $\rho_1$.
The algorithm recursively computes a \emph{DP table} for each edge of $T$ using bottom-up dynamic programming.
Each of those tables stores information about partial solutions computed on the subgraph lying inside the noose of the corresponding edge.
We have a small difference from \cite{DPBF10} in that we specifically choose the root noose to be of constant length, because the processing of it is slightly different from processing the internal nooses (\cite{DPBF10} instead pick an arbitrary edge in $T$ and append an additional fictitious root edge).
In this way we avoid dealing with the complexity of processing the root noose, as it will take constant time.

Let $e$ be an edge of $T$, and denote by $G_e$ the subgraph of $G$ corresponding to the subtree of $T$ bounded by the noose $O_e$.
A \emph{partial Hamiltonian cycle} of $G_e$ is a collection of disjoint paths (possibly consisting of a single vertex) in $G_e$ such that each vertex of $G_e$ is in some path and each of the paths has both endpoints in $\midset(e)$. \Cref{fig:hamiltonian} shows (left side) two adjacent nooses and (right side) their partial Hamiltonian Cycles.
Partial Hamiltonian cycles are used to construct, bottom-up, other partial Hamiltonian cycles of increasing size, with the purpose that at the root, they can be used to construct a full Hamiltonian cycle of $G$.
To compose a partial Hamiltonian cycle of $G_e$ with other partial Hamiltonian cycles, it is important only to retain the topology of its intersection with $\midset(e)$.
To model this topology we introduce labels. A \emph{labeling} of a vertex subset $X \subseteq V(G)$ is a mapping $s : X \to \Sigma$, where $\Sigma$ is the chosen set of labels.

For a partial Hamiltonian cycle $H$ of $G_e$, the labeling is defined as follows.
Each vertex of $\midset(e)$ is assigned one of the four symbols $\Sigma = \{\texttt{0},\texttt{2},\texttt{[},\texttt{]}\}$.
A vertex $v$ is assigned $\texttt{0}$ or $\texttt{2}$ if the degree of $v$ in $H$ is $0$ or $2$, respectively.
Otherwise, the vertex is assigned either the label \texttt{[} or \texttt{]}.
These labels define a non-crossing matching of the vertices of $\midset(e)$ labeled by the bracket symbols.
The cyclic ordering over the vertices is prescribed by $\pi(e)$.
\Cref{fig:hamiltonian} shows labelings for two adjacent nooses (right).
The arc of the matching between two vertices labeled by \texttt{[} and \texttt{]} corresponds to a path of~$H$ that connects them.
Between the two endpoint vertices, a path may pass through other vertices of $\midset(e)$, labeled by \texttt{2}.

For a labeling $s$ of $\midset(e)$, define the DP table value $D[e][s]$ as the minimum weight of a partial Hamiltonian cycle of $G_e$ compatible with $s$, or $+\infty$ if none exists.
Now we describe how to construct the DP table for an edge $e_P$, and denote $G_{e_P}$ by $G_P$ for simplicity.
If $e_P \neq e_{\rho}$ is an edge incident to a leaf of $T$, $O_P$ encloses a single edge $\{u,v\}$ of $G$, so the valid labelings are $(\texttt{0},\texttt{0})$ (with weight $0$) and $(\texttt{[},\texttt{]})$ (with weight $w(\{u,v\})$).
Otherwise, we look at its two child edges $e_L$ and $e_R$.
We assume that the algorithm has already computed the DP tables for $e_L$ and $e_R$.
Let the nooses of $e_P$, $e_L,$ and $e_R$ be $O_P$, $O_L$ and $O_R$, respectively.
By the definition of the sphere cut decomposition, the union of the disks $\Delta_L$ and $\Delta_R$ bounded by $O_L$ and $O_R$, respectively, is equal to the disk $\Delta_P$ bounded by $O_P$.
Following the notation of \cite{Dor06}, we define the following sets of vertices on these nooses (illustrated in \Cref{fig:hamiltonian} on the left):
\begin{enumerate*}[label=(\roman*)]
    \item the \emph{left} set $L = \midset(e_P) \setminus \midset(e_R)$;
    \item the \emph{right} set $R = \midset(e_P) \setminus \midset(e_L)$;
    \item the \emph{forget} (which we also call \emph{frontier}) set $F = (\midset(e_L) \cup \midset(e_R)) \setminus \midset(e_P)$;
    \item the \emph{intersection} set $I = \midset(e_L)\cap \midset(e_R) \cap \midset(e_P)$.
\end{enumerate*}

\begin{figure}[ht!]
    \centering
    \adjustbox{valign=c}{
    \begin{tikzpicture}[scale=1.2]
        \input{fig-nooses}
    \end{tikzpicture}
    }
    \hspace{1cm}
    \adjustbox{valign=c}{
    \begin{tikzpicture}[scale=1.2]
        \input{fig-labeling}
    \end{tikzpicture}
    }
    \caption{On the left, the left, right, forget and intersection sets $L$, $R$, $F$ and $I$.
    In green, the order of the labeling for $O_P$, and in blue, the order of the labelings on $O_L$ and $O_R$.
    On the right, a partial Hamiltonian cycle of $G_P$ and the corresponding vertex labelings for the nooses $O_L$ and $O_R$.
    The paths of the partial Hamiltonian cycle pass through all vertices of $G_P$, internal in $O_L$ and $O_R$, which are not depicted for clarity.
    In green, the labeling $s_P$ on the boundary on $O_P$, and in blue, feasible labelings $s_L$ and $s_R$ on $O_L$ and $O_R$.}
    \label{fig:hamiltonian}
\end{figure}
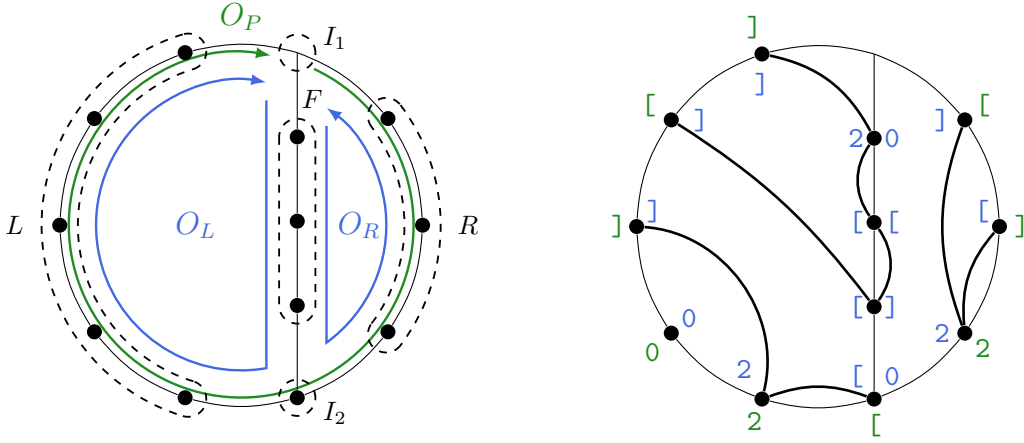

The nooses $O_L$ and $O_R$ enclose the subgraphs $G_L = G_{e_L}$ and $G_R=G_{e_R}$ of $G_P$, respectively.
Note that $|I| \leq 2$, since there are at most two points on the sphere where the noose $O_P$ intersects with both $O_L$ and $O_R$ ($I_1$ and $I_2$).
In the rest of the paper, we will ignore $I$ in the complexity analysis for simplicity, since it does not impact the asymptotic complexity. 
The left side of \Cref{fig:hamiltonian} graphically depicts the subsets $L$, $R$, $F$, and $I$.
In the example in the figure, there is a single intersection vertex.
We fix the clockwise ordering in $O_L$ and counter-clockwise ordering in $O_R$, both starting from the point $I_1$.
The right side of \Cref{fig:hamiltonian} shows a partial Hamiltonian cycle of $G_P$, which induces labelings $(\texttt{2},\texttt{[},\texttt{[},\texttt{[},\texttt{2},\texttt{0},\texttt{]},\texttt{]},\texttt{]})$ and $(\texttt{0},\texttt{[},\texttt{]},\texttt{0},\texttt{2},\texttt{[},\texttt{]})$ for $G_L$ and $G_R$, respectively.

The DP table for $D[e]$ is then constructed as follows.
We additionally call $B = \midset(e_P) = L \cup R \cup I$ the \emph{boundary}.
We first bruteforce all possible label assignments to the boundary and, for each such labeling $s_P$, we determine whether it agrees with some partial Hamiltonian cycle of $G_P$.
To check that, we proceed as follows.

First, we bruteforce all possible \emph{double} label assignments on the frontier and on the intersection vertices, i.e., each of these vertices is assigned a label on both sides of the intersection $F \cup I$.
Furthermore, we proceed with this double labeling only if at each vertex in $I$, its pair of labels is consistent with its label in $s_P$ (for example, a pair of labels $(\texttt{[},\texttt{0})$ is not consistent with the boundary labels \texttt{2} or \texttt{0}).
Then, when all the labels are fixed, we first check whether such a labeling is \emph{valid}, i.e., whether it can, in principle, define a feasible partial Hamiltonian cycle.
This is true if, firstly, the bracket sequences in each of the nooses $O_L$ and $O_R$ are correct.
Secondly, the labelings on the frontier must not induce a cycle.
Finally, the vertices with label $\texttt{2}$ can always be joined to some path without changing the topology (except only if the noose contains a label $\texttt{2}$ but no bracket labels).
The vertices with label $\texttt{0}$ can be ignored.
All these conditions can be checked in polynomial time in terms of $\ell$.

Second, if the labeling is valid, we proceed as follows. Note that the labeling at the frontier and at the intersection vertices together with the labels from $s_P$ for the vertices in $L$ and $R$ completely define a labeling $s_L$ of $O_L$ and a labeling $s_R$ of $O_R$, respectively (see \Cref{fig:hamiltonian} for an example).
We can check efficiently whether $G_L$ and $G_R$ admit partial Hamiltonian cycles corresponding to $s_L$ and $s_R$, respectively, by checking whether $D[e_L][s_L]\neq +\infty$ and $D[e_R][s_R]\neq +\infty$.
The table $D$ needs to be implemented as an efficient search data structure, e.g., a balanced binary tree.
If both tests succeed, we update $D[e_P][s_P]$ to be the minimum of its current value and $D[e_L][s_L]+D[e_R][s_R]$.

Finally, since we assumed that $r$ belongs to the Hamiltonian cycle, we look up the value of $D[e_\rho][s]$ for $s=(\texttt{[},\texttt{]})$, so $D[e_\rho][s] + w(r)$ is the corresponding solution candidate.

Before giving the time complexity of this algorithm for Planar TSP, we use this analysis to outline the general complexity result that is obtained by applying this dynamic programming strategy for an arbitrary connectivity problem.
For a given problem, suppose an appropriate noose labeling scheme has been chosen.
Denote the number of labeling options considered on the boundary and on the frontier by $\mathcal B$ and $\mathcal F$, respectively (we will use this notation in all of our algorithms further on), and assume that the corresponding labelings can be generated in $\Oh^*(\mathcal B)$ and $\Oh^*(\mathcal F)$ time, respectively.
Define the constants $c_1 = \mathcal B^{1/|B|}$ and $c_2 = \mathcal F^{1/|F|}$.
Then we have the following general theorem, which is implicit in \cite{DPBF10}.

\begin{theorem} \label{thm:general-dp}
    Assume that a problem on a weighted graph can be solved by the above deterministic dynamic programming algorithm for a certain labeling scheme with constants $c_1$ and $c_2$, where $c_1^2 \geq c_2$.
    The complexity of the algorithm then is
    \begin{equation}
        \Oh^*((c_1\sqrt{c_2})^{\bw}).
    \end{equation}
\end{theorem}

We give a very short proof which considerably simplifies the analysis of \cite{DPBF10}.

\begin{proof}
The complexity of an iteration of the algorithm can be calculated by examining the total number of labelings.
This number is at most 
\begin{equation}
    \mathcal{B\cdot F} = c_1^{|B|} c_2^{|F|} = \Oh(c_1^{|L|+|R|}c_2^{|F|}).
\end{equation}

By the definition of the branch decomposition, with $\ell=\bw$, we have $|L|+|F| \leq \ell$ and $|R|+|F|\leq \ell$.
Therefore, if $|L|+|R| = \ell-2t$ for some $t \geq 0$, then $|F| \leq \ell/2 + t$.
Then the total number of options is asymptotically upper bounded by
\begin{align}
    c_1^{|L|+|R|} \cdot c_2^{|F|} = c_1^{\ell-2t} c_2^{\ell/2 + t} = c_1^{\ell} c_2^{\ell/2} \left(\frac{c_2}{c_1^2}\right)^t \leq (c_1\sqrt{c_2})^\ell.
\end{align}
Therefore, the worst case is achieved when $t = 0$ and $|F| = \ell/2$, $|L|+|R| = \ell$.
Since the size of $T$ is $\Oh(n)$, the total running time for the dynamic programming is $\Oh^*((c_1\sqrt{c_2})^\ell)$.
\end{proof}

This theorem implies the following result for the \textsc{Planar Travelling Salesman} problem.

\begin{theorem}[{\cite[Lemma 1]{DPBF10}}] \label{thm:tsp-dp}
    There exists a deterministic algorithm that solves the \textup{\textsc{Planar Travelling Salesman} problem} in $\Oh^*((4\sqrt 6)^{\bw}) = \Oh(2^{6.987\sqrt{n}})$ time.
\end{theorem}

\begin{proof}
The label (or a pair of labels) for a vertex is picked independently of the others, so the generation of the labelings has the same asymptotic complexity as the number of labelings.

On the boundary, we have $\mathcal B = 4^{|B|} = \Oh(4^{|L|+|R|})$.
On the frontier, there are $6$ possible label pair options for each vertex, $(\texttt{0},\texttt{2}),(\texttt{2},\texttt{0}),(\texttt{[},\texttt{[}),(\texttt{[},\texttt{]}),(\texttt{]},\texttt{[})$, and $(\texttt{]},\texttt{]})$, hence $\mathcal F \leq 6^{|F|}$.
Thus, $c_1=4$ and $c_2=6$ and the bound in the statement follows.    
\end{proof}

Dorn showed that one can use fast matrix multiplication to obtain the following improvement \cite{Dor06}, see~\cref{app:fmm} where we provide more background on fast matrix multiplication and revisit the proof of the following theorem for completeness.

\begin{theorem}[{\cite[Theorem 4]{Dor06}}] \label{thm:general-fmm}
    Assume that a problem on a weighted graph with weights at most $W$ can be solved by the above deterministic dynamic programming algorithm for an appropriate labeling scheme with constants $c_1$ and $c_2$, where $c_1 \leq c_2$.
    The complexity of the algorithm then is
    \begin{equation}
        \Oh^*\Big(W\cdot \max(c_1^{\frac{\omega-1}{2}} \sqrt{c_2}, c_2)^{\bw}\Big).
    \end{equation}
\end{theorem}

Compared to \Cref{thm:general-dp}, the improvement is in the reduction of the constant in the power of $c_1$ from $1$ to $(\omega-1)/2 < 0.686$.
Here, $\omega < 2.372$ is the fast matrix multiplication constant \cite{ADWXXZ25}.
However, this comes with an additional factor of $W$ in the complexity.

This can be leveraged to obtain the following result:
\begin{theorem} \label{thm:tsp-fmm}
    There exists a deterministic algorithm that solves the \textup{\textsc{Planar Travelling Salesman} problem} in $\Oh^*(W\cdot (4^{\frac{\omega-1}{2}} \sqrt{6})^{\bw}) = \Oh(W\cdot 2^{5.655\sqrt{n}})$ time.
\end{theorem}

\subsection{Speedup by Meandric Systems} \label{sec:speedup}

Let us focus again on the \textsc{Planar Travelling Salesman} problem.
With the above labeling scheme, we have $\mathcal F = 6^{|F|}$.
The quantity of the total number of labelings to be considered can be improved by excluding some \emph{invalid} labelings.
In particular, we can exclude the frontier labeling pairs $s = (s_L,s_R)$ that induce a cycle, since they cannot be completed to a Hamiltonian cycle of $G$; see, e.g., \cref{fig:systems} (right).

The number of labelings $\mathcal F$ can be calculated in the following way: first pick a subset of vertices $S \subseteq F$ for which to assign any of the $4$ pairs of labels $(\texttt{[}, \texttt{]})$, $(\texttt{[}, \texttt{[})$, $(\texttt{]}, \texttt{]})$,  and $(\texttt{]}, \texttt{[})$.
Each of the other $F \setminus S$ vertices is assigned either the pair of labels $(\texttt{0},\texttt{2})$ or $(\texttt{2}, \texttt{0})$.
The two bracket sequences in $S$ then should define a meandric system.
As mentioned in~\cref{se:intro}, Bobier and Sawada were the first to notice that the valid frontier configurations are exactly the unidirectional open acyclic meandric systems \cite{BS10}.
The unclosed arcs of these systems extend only to the right, since (see \Cref{fig:hamiltonian}) the labelings of $O_L$ and $O_R$ start at $I_1$.
The acyclicity implies no separate disconnected cycles.
Let $U_m$ be the number of such meandric systems on $m=|S|$ crossings and let $z$ be the number such that $U_m = \Theta^*(z^m)$.

Trivially, $z \leq 4$, as each crossing has $2$ bracket labels on both sides.
In \cite{DPBF10}, the authors already gave a custom upper bound of $z \leq 3.68133$, without explicitly relating the configurations to the meandric systems.
They proved it by counting the number of different classes of frontier labelings by iteratively multiplying a small matrix, and analyzing the largest eigenvalue of this matrix.
This improves the asymptotic upper bound on $\mathcal F$:
\begin{equation}
    \sum_{S \subseteq F} z^{|S|} 2^{|F|-|S|} = \sum_{i=0}^{|F|} \binom{|F|}{i} z^i 2^{|F|-i} =
    2^{|F|} \cdot \left(1+\frac{z}{2}\right)^{|F|} = (2+z)^{|F|}.
\end{equation}

Therefore, $c_2 = 2+z < 5.682$.
Hence, for the \textsc{Planar Travelling Salesman} problem, by plugging this value into
\Cref{thm:tsp-dp}, the authors of \cite{DPBF10} obtain an algorithm having $\Oh^*(9.535^{\bw}) = \Oh(2^{6.904\sqrt n})$ time complexity, and by plugging this value into \Cref{thm:tsp-fmm} the work \cite{Dor06}, using fast matrix multiplication, obtains an algorithm having $\Oh^*(W\cdot 6.170^{\bw}) = \Oh(W\cdot 2^{5.571\sqrt{n}})$ time complexity.

Bobier and Sawada have numerically calculated the exact values of $U_m$ up to $m = 28$ \cite{BS10}.
However, the authors did not improve the upper bound on the value of $z$.
\jvchange{As a warm-up, first} we point out a marginally better bound that we can get by only using their results.
They have also calculated the exact numbers $S_m$ of \emph{bidirectional} open acyclic meandric systems up to $m = 28$ (however, $S_{27}$ is claimed to be incorrect by \cite{oeis:A000108}).
Relying on the largest correct value $S_{26} = 434578910451203$, we can partition the $m$ vertices of the unidirectional acyclic system into $\Oh(m/26)$ consecutive groups of size at most $26$, let then the labeling in the $i$-th part be $s^{(i)} = (s^{(i)}_L,s^{(i)}_R)$.
Then, the meandric system defined by $s^{(i)}$ is bidirectional and acyclic.
Hence,
\begin{equation}
    z \leq S_{26}^{1/26} = \sqrt[26]{434578910451203} \leq 3.656.
\end{equation}

\jvchange{Next, we show that a stronger upper bound follows from the known upper bound on the meandric numbers $M_n$ \cite{AP05}.}
\begin{theorem} \label{thm:z-bound}
$z \leq 3.592$.
\end{theorem}

\begin{proof}
    To obtain the best currently-known upper bound on the number of meanders $M_n$ of order~$n$, Albert and Paterson used the \emph{\jvchange{Goulden-Jackson} cluster method} \cite{GJ83}.
    For an alphabet $\Sigma$ and a set $\mathcal W \subset \Sigma^*$, a $\mathcal W$-factor-free word in $\Sigma^*$ is a word that does not contain any word from $\mathcal W$ as a substring.
    Formally, the language of $\mathcal W$-factor-free words is $\mathcal L_{\mathcal W} = \Sigma^* \setminus (\Sigma^* \mathcal W \Sigma^*)$.
    The cluster method can be used to obtain an asymptotic estimate on the number of words in $\mathcal L$ of fixed length.

    The authors encode meandric systems with the alphabet $\Sigma = \{ \texttt{U}, \texttt{D}, \texttt{L}, \texttt{R} \}$.
    Each of these symbols describes the labeling on both sides of a crossing.
    This corresponds to our bracket labeling as $\texttt{U} = (\texttt{[},\texttt{[})$, $\texttt{D} = (\texttt{]},\texttt{]})$, $\texttt{L} = (\texttt{[},\texttt{]})$, and $\texttt{R} = (\texttt{]},\texttt{[})$.
    \jvchange{ Let $\mathcal M \subset \Sigma^*$ be the set of words that correspond to words containing a closed submeander which passes through the first and the last crossing (so the word is of the form \texttt{U}\ldots\texttt{D}).
    Other strands may intersect the interior of the horizontal line, for example, $\texttt{ULLRRD} \in \mathcal M$.
    }
    
    The language $\texttt{U} \mathcal L_{\mathcal M} \texttt{D}$ contains $\mathcal M$, as meanders do not contain a disconnected loop. 
    An upper bound on the number of length-$n$ words in $\texttt{U} \mathcal L_{\mathcal M} \texttt{D}$ then also gives an upper bound on the number of length-$n$ words in~$\mathcal M$.
    In this way, they were able to show an upper bound of $M_n = \Oh^*(12.901^n)$.
    Instead of the whole language $\mathcal M$, they found a representation of a smaller set $\mathcal M' \subset \mathcal M$ by a computer program and forbade the words from $\mathcal M'$. 
    Now note that $\mathcal L_{\mathcal M}$ includes the language of bidirectional acyclic meandric systems, since all that was done was to forbid the submeanders.
    Since each unidirectional meandric system is also bidirectional and since $m=2n$, we have $U_m \leq S_m = \Oh^*(\sqrt{12.901^m}) = \Oh^*(3.592^m)$, and the statement follows.\end{proof}

\paragraph{Improved Upper Bound on the Meandric Numbers.}

The upper bound $M_n  = \Oh^*(12.901^n)$ is not tight, as already mentioned in \cite{AP05}, and the authors conjectured that the techniques could be pushed further.
By considering larger sets $\mathcal M'$, the upper bound provided by the Goulden-Jackson cluster method is certain to be tighter.
Considering that the work \cite{AP05} was developed over 20 years ago, it is clear that we can obtain an improved upper bound by running their algorithm on a computer with modern resources.

We do not describe the algorithm itself here, as it is already laid out in \cite{AP05}, but we provide some intuition.
The language of forbidden words $\mathcal M'$ is a finite language which can be constructed as follows.
First, the algorithm generates a set of normalized forbidden words of $\mathcal M'$ of fixed length, called \emph{standard representatives}, using backtracking.
Each such word contains a meander as a factor.
A fragment is called a \emph{shift} if its effect on a vertical cross-section of a meandric system is a vertical shift either upwards or downwards by some offset (for example, \texttt{ULD} or \texttt{LRR}).
The standard representatives are minimal, in the sense that each shift factor of such words is in its simplest form (i.e., consist either only of \texttt{L} or only of \texttt{R} symbols).
Each word in $\mathcal M'$ can be obtained from one standard representative by substituting shifts with more complex shifts.
Each complex shift in turn consists of \emph{jumps}, which are simplest shifts.
The jumps are also generated by backtracking, and then are used to build shifts.

Instead of simply generating the words in $\mathcal M'$, the algorithm builds generating functions for standard representatives that describe all the words that can arise from each of them.
These are then used for the Goulden-Jackson cluster method, which constructs a linear system of equations.
This system is then solved iteratively and numerically; \cite{AP05} show that such a process converges to the (unique) solution.
This solution is then used to construct the generating polynomial of $\mathcal M'$, the convergence radius of which defines the corresponding upper bound.

To implement the procedure, we have used ChatGPT-5.6 Sol.
The model was able to implement the algorithm, along with a few algorithmic optimizations.
The full code and results are available at \cite{BLCV26}.
By increasing the length of the standard representatives $n$, we obtain better upper bounds.
Albert and Paterson have calculated the upper bound up to $n = 16$, which is $12.901$.
Using our program, we have calculated the values up until $n = 24$, which gives $12.805332\ldots$
Thus, we have an improved upper bound on $z$:

\upperbound*

\Cref{tab:meandric-numbers} lists our upper bounds for various values of $n$.
The results of our program are slightly better than \cite{AP05}, which can be explained by more extensive generation of words (longer jumps, shifts are constructed of a larger number of jumps, more iterations etc.).
The values are close, but exact reproduction of their results cannot be made because \cite{AP05} do not describe all the parameters of their calculation.
\begin{table}[ht!]
    \centering
    \begin{tabular}{c|c|c}
        $n$ & This paper & \cite{AP05} \\ \hline
        2 & 13.511115 & ---\\
        4 & 13.351046 & ---\\
        6 & 13.171304 & 13.171 \\
        8 & 13.084026 & 13.086 \\
        10 & 13.015332 & 13.018 \\
        12 & 12.965070 & 12.970 \\
        14 & 12.925273 & 12.931 \\
        16 & 12.892914 & 12.901 \\
        18 & 12.865872 & ---\\
        20 & 12.842819 & ---\\
        22 & 12.822854 & ---\\
        24 & 12.805332 & ---
    \end{tabular}
    \caption{Upper bounds on the meandric numbers.}
    \label{tab:meandric-numbers}
\end{table}

Our results confirm the suggestion of \cite{AP05} that the upper bound can be pushed to around $12.8$.
However, we also see that the convergence of the upper bound is still not very fast; the obtained upper bounds closely follow the curve $12.34 + 1.8/n^{0.426}$, as we can see in \Cref{fig:extrapolation}.
Proving a similar upper bound with the current method would be infeasible, since it would require $n$ to be of order of several thousands.
Such an upper bound would give $z \leq 3.513$.

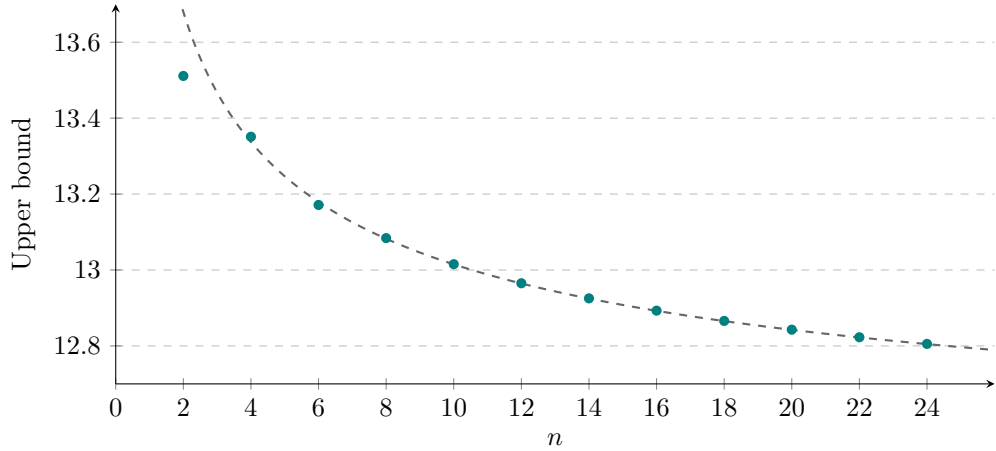
\begin{figure}[ht!]
\begin{center}
\input{graph.tex}
\end{center}
\caption{Upper bounds for the meandric number $M$ depending on the length of the standard representatives $n$ for the Goulden-Jackson cluster method.
The dashed approximating curve is the function $12.34 + 1.8/n^{0.426}$.}
\label{fig:extrapolation}
\end{figure}

\paragraph{\jvchange{Improved Complexity Upper Bounds for Planar TSP.}}

\jvchange{Let us now return to the \textsc{Planar Travelling Salesman} problem.}
Note that the list of unidirectional acyclic meandric systems of length $m$ can be generated in $\Oh^*(U_m) = \Oh^*(z^m)$ time as follows.
Start with an empty string
and generate labeling pairs for the meandric system using backtracking, each time trying to label the last unlabeled vertex with all of the $4$ options $(\texttt{[},\texttt{[})$, $(\texttt{]},\texttt{]})$, $(\texttt{[},\texttt{]})$, and $(\texttt{]},\texttt{[})$.
Then check if the new labels create a contradiction or a closed disjoint loop in $\Oh(\poly m)$ time.
If not, then continue to generate the labeling recursively.

The time complexity of this algorithm is still $\Oh^*(U_m)$, because for any valid prefix of the labeling, if it contains no contradictions or closed loops, labeling all the remaining vertices with $(\texttt{[},\texttt{[})$ creates a valid unidirectional acyclic meandric system.
Hence, the size of the backtracking tree is $\Oh^*(U_m)$, where at each step we perform some $\Oh(\poly m)$-time computation.
Then we have the following~result.

\begin{theorem} \label{thm:tsp-meanders}
    There exist deterministic algorithms that solve the \textup{\textsc{Planar Travelling Salesman} problem} in $\Oh^*((4 \sqrt{2+z})^{\bw}) = \Oh(2^{6.876\sqrt{n}})$ and $\Oh^*(W\cdot (4^{\frac{\omega-1}{2}} \sqrt{2+z})^{\bw}) = \Oh(W\cdot 2^{5.543\sqrt{n}})$ time.
\end{theorem}

This improves the best-known complexity $\Oh^*(W\cdot 2^{5.571\sqrt{n}})$ for the \textsc{Planar Travelling Salesman} problem in the bounded-weight setting by \cite{Dor06} (updated with the improved $\omega \leq 2.372$).
In the unbounded-weight setting, the best complexity is given by a more recent work improving the complexity of some planar connectivity problems, which is $\Oh(2^{6.570\sqrt{n}})$ \cite{PBvR17}, which our classical algorithm does not improve.

As a corollary, our theorem also improves the best-known complexity of the \textsc{Planar Hamiltonian Cycle}, since it is just the unweighted \textsc{Planar Travelling Salesman} problem.
\begin{theorem} \label{thm:hc-meanders}
    There exists a deterministic algorithm that solves the \textup{\textsc{Planar Hamiltonian Cycle}} problem in $\Oh^*((4^{\frac{\omega-1}{2}} \sqrt{2+z})^{\bw}) = \Oh(2^{5.543\sqrt{n}})$ time.
\end{theorem}

Finally, we mention that simple modifications of the algorithms imply the same complexity for the \textsc{Planar Hamiltonian Path} problem (for example, see \cite[Section 5]{DPBF10}).

\subsection{Planar Longest Cycle/Path}

Similarly to \cite{DPBF10}, we can adopt the previous framework, slightly changing the algorithm that we have used to solve the \textsc{Planar Travelling Salesman} problem, to tackle the \textsc{Planar Longest Cycle} problem (the \textsc{Planar Longest Path} problem follows similarly as in~\cite{DPBF10}).

The restrictions on the meandric systems here are the same as we have for the \textsc{Planar Travelling Salesman} problem.
The main difference is that instead of processing partial Hamiltonian cycles, we process partial cycles, allowing some vertices to be left out of the partial solution.
Such vertices on the frontier will have labels $(\texttt{0}, \texttt{0})$.
In the complexity analysis, this has the consequence of $c_2$ changing from $2+z$ to $3+z$. Hence, we obtain the following.

\begin{theorem} \label{thm:plc-meanders}
    There exist deterministic algorithms that solve the \textup{\textsc{Planar Longest Cycle/Path}} problem in $\Oh^*((4 \sqrt{3+z})^{\bw}) = \Oh(2^{7.128\sqrt{n}})$ and $\Oh^*(W\cdot (4^{\frac{\omega-1}{2}} \sqrt{3+z})^{\bw}) = \Oh(W\cdot 2^{5.796\sqrt{n}})$ time.
\end{theorem}

The previous results by Dorn, Penninkx, Bodlaender and Fomin achieved complexities $\Oh(2^{7.152\sqrt{n}})$ and $\Oh(W\cdot 2^{5.819\sqrt{n}})$, adjusted for the current-best value of $\omega$ \cite{Dor06,Dor07,DPBF10}.
These were also obtained by excluding forbidden configurations, which implicitly correspond to meandric systems, with the upper bound $z \leq 3.682$.
Our results give a tighter bound on their algorithms.

In the next sections, we use meanders to develop new algorithms with improved complexity for other problems. We also examine another type of meandric systems for which the bound of \Cref{thm:z-bound} does not hold, and apply a simplified version of the cluster method to derive an upper bound on the number of such systems.

\section{Tree Optimization Problems}

In this section, we examine planar problems in which we need to find an optimal subtree of the given planar graph $G$ that satisfies some criteria.
We illustrate our algorithms for solving the \textsc{Planar Steiner Tree} and \textsc{Planar Connected Dominating Set} problems.

We adapt the method of excluding meandric systems to the search for an optimal subtree. It is worth noting that, in contrast to the cycle optimization problems of the previous section, the existing tree optimization algorithms did not really use the meandric systems to exclude certain configurations to~improve their~running~time.

\subsection{Eulerian Outline}

To use the meandric systems in such tree optimization problems, we introduce the following notion of a topological outline of a tree on a sphere.
Consider a spherical drawing of $G$.
Let $\varepsilon$ be a small constant such that the geodesic distance between any two points on two edges with no common endpoints is greater than $2\varepsilon$.
For a subtree $S$ of $G$, let $\mathcal E(S)$ be the set of points $P$ such that the minimum geodesic distance from $P$ to the edges of $S$ is equal to $\varepsilon$.
We call $\mathcal E(S)$ the \emph{Eulerian outline} of $S$.

Essentially, $\mathcal E(S)$ is the perimeter of $S$ offset at distance $\varepsilon$; see \Cref{fig:eulerian} (left).
This outline corresponds to the Eulerian tour of $S$, where each edge is traversed twice, once in each direction.
More formally, it is defined as the boundary of the Minkowski sum of the drawing of the tree and a disk of radius $\varepsilon$, called $\varepsilon$-offset in \cite{BHKP12}.
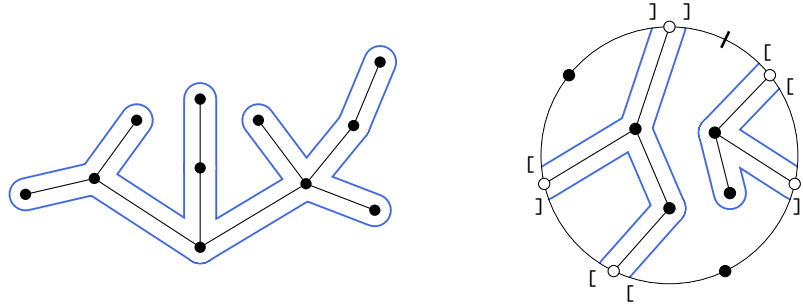
\begin{figure}[ht!]
    \centering
    \adjustbox{valign=c}{
    \begin{tikzpicture}[
        scale=.7,
        vertex/.style={circle, fill=black, inner sep=1.5pt},
        tube/.style={line cap=round, line join=round}
    ]
    
    \coordinate (v1)  at (0,0);
    \coordinate (v2)  at (-2,1.3);
    \coordinate (v3)  at (0,1.5);
    \coordinate (v4)  at (2,1.2);
    \coordinate (v5)  at (-3.3,1);
    \coordinate (v6)  at (-1.2,2.4);
    \coordinate (v7)  at (0,2.8);
    \coordinate (v8)  at (1.1,2.4);
    \coordinate (v9)  at (2.9,2.3);
    \coordinate (v10) at (3.4,3.5);
    \coordinate (v11) at (3.3,.7);
    
    \def\treepath{
        (v1) -- (v2)
        (v1) -- (v3)
        (v1) -- (v4)
        (v2) -- (v5)
        (v2) -- (v6)
        (v3) -- (v7)
        (v4) -- (v8)
        (v4) -- (v9)
        (v9) -- (v10)
        (v4) -- (v11)
    }
    
    \draw[tube, line width=0.45cm, draw=RoyalBlue] \treepath;
    
    \draw[tube, line width=0.4cm, draw=white] \treepath;
    
    \draw \treepath;
    
    \foreach \p in {v1,v2,v3,v4,v5,v6,v7,v8,v9,v10,v11}
        \node[vertex] at (\p) {};
    
    \end{tikzpicture}
    }
    \hspace{1cm}
    \adjustbox{valign=c}{
    \begin{tikzpicture}
        [tube/.style={line cap=round, line join=round}]
        \tikzstyle{blackv} = [circle, draw, fill=black, inner sep=1.5pt]
        \tikzstyle{whitev} = [circle, draw, fill=white, inner sep=1.5pt]

        \def\r{1.7}
    
        \foreach \i in {1,...,7} {
            \coordinate (v\i) at ({90 + (\i-1)*360/7}:\r);
        }

        \coordinate (A) at (-0.45,0.35);
        \coordinate (B) at (0,-0.7);
        \coordinate (C) at (0.6,0.3);
        \coordinate (D) at (0.8,-0.5);
        
        \def\tpa{
            (v1) -- (A)
            (A) -- (v3)
            (A) -- (B)
            (B) -- (v4)
        }
    
        \draw[tube, line width=0.45cm, draw=RoyalBlue] \tpa;
        \draw[tube, line width=0.4cm, draw=white] \tpa;
        
        \def\tpb{
            (v7) -- (C)
            (C) -- (D)
            (C) -- (v6)
        }
    
        \draw[tube, line width=0.45cm, draw=RoyalBlue] \tpb;
        \draw[tube, line width=0.4cm, draw=white] \tpb;
    
        \draw (0,0) circle (\r);

        \node[blackv] at (A) {};
        \node[blackv] at (B) {};
        \node[blackv] at (C) {};
        \node[blackv] at (D) {};
    
        \draw (v1) -- (A);
        \draw (A) -- (B);
        \draw (A) -- (v3);
        \draw (B) -- (v4);
    
        \draw (v7) -- (C);
        \draw (C) -- (v6);
        \draw (C) -- (D);
        
        \draw[line width=0.38cm, white] (0,0) circle (\r+0.2);
    
        \foreach \i in {1,3,4,6,7} {
            \node[whitev] at (v\i) {};
        }
        \foreach \i in {2,5} {
            \node[blackv] at (v\i) {};
        }

        \scriptsize
        \node[xshift=-.05cm,yshift=.3cm] at (v7) {$\texttt{[}$};
        \node[xshift=.25cm,yshift=-.15cm] at (v7) {$\texttt{[}$};
        \node[xshift=.2cm,yshift=.25cm] at (v6) {$\texttt{]}$};
        \node[xshift=.05cm,yshift=-.3cm] at (v6) {$\texttt{]}$};
        \node[xshift=.2cm,yshift=-.3cm] at (v4) {$\texttt{[}$};
        \node[xshift=-.3cm,yshift=-.1cm] at (v4) {$\texttt{[}$};
        \node[xshift=-.05cm,yshift=-.3cm] at (v3) {$\texttt{]}$};
        \node[xshift=-.2cm,yshift=.25cm] at (v3) {$\texttt{[}$};
        \node[xshift=-.2cm,yshift=.2cm] at (v1) {$\texttt{]}$};
        \node[xshift=.25cm,yshift=.2cm] at (v1) {$\texttt{]}$};

        \draw[line width=1pt] ({450/7}: {\r-0.12}) -- ({450/7}: {\r+0.12});
    \end{tikzpicture}
    }
    
    \caption{On the left, the Eulerian outline of a spherical drawing of a tree.
    On the right, a balanced bracket labeling of a partial Eulerian outline (the order starts from the small segment cutting an arc in the top right and continues in the clockwise direction).
    The vertices not belonging to the subtree are not labeled.}
    \label{fig:eulerian}
\end{figure}

Consider a noose $O_e$ of a sphere cut decomposition of $G$ and the intersection of its enclosed disk~$\Delta_e$ with $\mathcal E(S)$.
If $\mathcal E(S)$ intersects with $O_e$, then this partial Eulerian outline consists of a set of Jordan arcs, where each arc has two endpoints on $O_e$.
If the number of vertices of $S$ in $\midset(e)$ is~$m$ (white-filled vertices in \cref{fig:eulerian} (right)), this naturally defines a non-crossing partition on these~$m$ vertices.

This non-crossing partition then can be represented by a balanced bracket labeling with $2m$ brackets, see \Cref{fig:eulerian} (right).
If a vertex $v$ is in $S$, then the Eulerian outline intersects $O_e$ at two points, one to the left and one to the right of $v$, which we call $\varepsilon$-points.
Each $\varepsilon$-point is then assigned one of the brackets \texttt{[} or \texttt{]}.

\subsection{Planar Connected Dominating Set}
\label{sec:cds}
We illustrate the use of the Eulerian outline combined with meandric systems to solve the \textsc{Planar Connected Dominating Set} problem.
We exploit the fact that this problem is equivalent to the \textsc{Planar Max Leaf} problem: for a given vertex-weighted planar graph, find a spanning tree $S$ with the largest total weight on its leaves (see e.g.~\cite{CWY00}).

In this case, we label the vertices on the noose $O_e$ with labels in $\{\texttt{0}_{\text{in}}, \texttt{0}_{\text{out}}, (\texttt{[},\texttt{[}), (\texttt{[},\texttt{]}), (\texttt{]},\texttt{[}), (\texttt{]},\texttt{]})\}$.
The label \texttt{0} means that the labeled vertex $v$ is a leaf in $S$; the subscripts in and out mean that~$v$ is connected to $S$ from inside or from outside of $O_e$, respectively.
The bracket pair labels are assigned to the non-leaf vertices.
The whole sequence of brackets then defines the topology of the partial Eulerian outline of $S$ inside $O_e$.
Note that, technically, the vertices labeled by \texttt{0} also are outlined by $\mathcal E(S)$, but we can ignore them in assigning the brackets, since they are disconnected in $S$ either from the subgraph inside $O_e$ or from the subgraph outside of $O_e$.
Intuitively, we can imagine that the vertices labeled by $\texttt{0}_{\text{in}}$ are pushed inside $O_e$ and the vertices labeled by $\texttt{0}_{\text{out}}$ are pushed outside of $O_e$ so that $\mathcal E(S)$ does not intersect $O_e$ in the radius $\varepsilon$ around them.

Now we can calculate the constants $c_1$ and $c_2$ to apply \cref{thm:general-dp,thm:general-fmm}.
\begin{itemize}
    \item On the boundary, each vertex can be assigned one of the six label options.
    Hence, $c_1 = 6$.
    \item On the frontier, we can count the number of options as follows.
    First, we bruteforce the set $A$ of vertices in $F$ that will be internal vertices in $S$.
    The vertices in $F \setminus A$ then will be leaves and labeled by \texttt{0}.
    The vertices in $A$ must be assigned a pair of brackets on each side of the frontier.
    We do not bruteforce all possible $4^{2m}$ bracket labelings but use meandric systems.

    Examine the meandric system defined by the bracket double labeling $s$ on $2m = 2|A|$ points.
    A closed loop (submeander) is not allowed, since it means one of the two things:
    \begin{itemize}
        \item If there is a submeander that starts at a point at an odd position of the labeling and ends at an even position, then there is a disconnected subgraph in $S$; see~\cref{fig:CDS-no-instance} blue submeander.
        \item If there is a submeander that starts at a point at an even position of the labeling and ends at an odd position, then it induces a cycle in $S$; see~\cref{fig:CDS-no-instance} orange submeander.
    \end{itemize}

    The first is not allowed, since we are looking for a spanning tree.
    The second is not allowed too, since $S$ must be a tree.
    Therefore, we need to examine only acyclic meandric systems on $2m$ points, the number of which is bounded by $\Oh^*(z^{2m})$.

    Therefore, in total, we examine $\Oh^*(z^{2|A|})$ labelings on $A$.
    The total number of labelings considered is
    \begin{equation}
        \mathcal F = \Oh^*\Bigl(\sum_{A \subseteq F} z^{2|A|}2^{|F|-|A|}\Bigr) =  \Oh^*\left((2+z^2)^{|F|}\right).
    \end{equation}
    Hence, $c_2 = 2+z^2 < 14.806$.
\end{itemize}
    
\begin{figure}[ht!]

    \centering
    \includegraphics[page=5,width=0.4\textwidth]{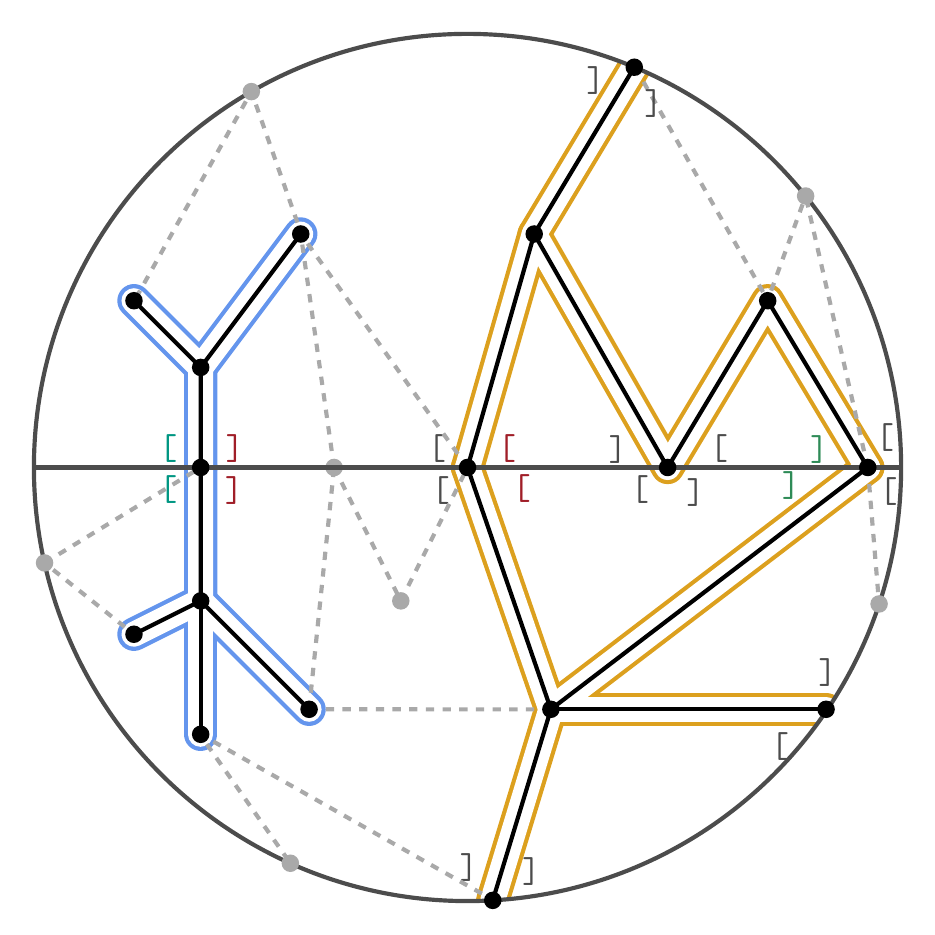}

\caption{ 
Illustration of forbidden meandric configurations for the {\sc Connected Dominating Set} problem. The blue submeander (left) starts at a point at an odd position of the labeling (green labels) and ends at an even position (red labels), hence it identifies a disconnected component. The orange submeander (right) starts at a point at an even position of the labeling (red labels) and ends at an odd position (green labels), hence it identifies a cycle. All other labels are colored in dark grey.}
\label{fig:CDS-no-instance}
\end{figure}

We obtain the following result:
\begin{theorem} \label{thm:cds}
    There exist deterministic algorithms that solve the \textup{\textsc{Planar Connected Dominating Set}} problem in $\Oh^*((6 \sqrt{2+z^2})^{\bw}) = \Oh(2^{9.611\sqrt{n}})$ and $\Oh^*(W\cdot (2+z^2)^{\bw}) = \Oh(W\cdot 2^{8.251\sqrt{n}})$ time.
\end{theorem}
Note that in the complexity of the second algorithm, $c_2$ dominates $c_1^{\frac{\omega-1}{2}} \sqrt{c_2}$.
The previous best-known complexities for deterministic algorithms are $\Oh(2^{8.778\sqrt{n}})$ by \cite{PBvR17} and $\Oh^*(W\cdot 2^{8.111\sqrt{n}})$ by \cite{Dor06}.
However, the latter complexity seems to be incorrect, since it corresponds to the~$c_1^{\frac{\omega-1}{2}} \sqrt{c_2}$ part, but the dominating part in their algorithm also comes from $c_2$, which for them is equal to~$c_2 = 17$, see \Cref{app:mistakes}.
This means that the correct upper bound on the complexity of their algorithm is $\Oh^*(W\cdot 2^{8.674\sqrt{n}})$.
Therefore, our result improves the best-known complexity in the bounded-weight setting.

\subsection{Planar Steiner Tree}

We can also apply meandric systems to the \textsc{Planar Steiner Tree} problem.

In this case, we use the following labeling.
If the vertex is not in $S$, we label it by \texttt{0}.
Otherwise, it is labeled by a pair of brackets, which correspond to intersections of the Eulerian outline of $S$ at the left and right $\varepsilon$-points of the vertex.
We can calculate the relevant quantities to apply \Cref{thm:general-dp} and \Cref{thm:general-fmm}.
\begin{itemize}

    \item Since each vertex labeled from the set $\{\texttt{0}, (\texttt{[},\texttt{[}), (\texttt{[},\texttt{]}), (\texttt{]},\texttt{[}), (\texttt{]},\texttt{]})\}$, we have $c_1 = 5$.
    
    \item On the frontier, we consider the possible options as follows.
    Bruteforce the set of vertices $A \subseteq F$ that are included in the Steiner tree.
    All other vertices then have labels \texttt{0} on both sides of the frontier.
    The $m = |A|$ vertices in $A$ are labeled by two pairs of brackets on each side of the frontier.
    This defines a double labeling $s = (s_L,s_R)$ on $2m$ vertices that may correspond to a partial Eulerian labeling of a subtree.

    First, we check whether all terminals $X$ are strictly contained inside $O_P$ (the parent noose that contains $O_L$ and $O_R$).
    If so, then we bruteforce all possible meanders on $A$ and check, using the DP tables of $O_L$ and $O_R$, whether we obtain a solution.
    In total, this contributes $M_m = \Oh^*(z^{2m})$ labelings.

    Otherwise, a closed loop on the frontier is not allowed by the same arguments as for the \textsc{Planar Connected Dominating Set} problem.
    By a similar calculation, we obtain
    \begin{equation}
        \mathcal F = \Oh^*\Bigl(\sum_{A \subseteq F} z^{2|A|}\Bigr) =  \Oh^*\left((1+z^2)^{|F|}\right).
    \end{equation}
    Therefore, $c_2 = 1+z^2 < 13.806$.
\end{itemize}

By applying \Cref{thm:general-dp} and \Cref{thm:general-fmm}, we obtain the following theorem:
\begin{theorem} \label{thm:steiner}
    There exist deterministic algorithms that solve the \textup{\textsc{Planar Steiner Tree}} problem in $\Oh^*((5 \sqrt{1+z^2})^{\bw}) = \Oh(2^{8.946\sqrt{n}})$ and $\Oh^*(W\cdot (1+z^2)^{\bw}) = \Oh(W\cdot 2^{8.037\sqrt{n}})$ time.
\end{theorem}

The papers \cite{DPBF10,Dor06} claim complexities $\Oh(2^{8.288\sqrt{n}})$ and $\Oh(W\cdot 2^{7.164\sqrt{n}})$; however, we point out a mistake in their proofs, see \Cref{app:mistakes}.
Fixing that results in larger complexities $\Oh(2^{9.264\sqrt{n}})$ and $\Oh(W \cdot 2^{8.674\sqrt{n}})$.
Using the improved upper bound for $z$, our result improves these older algorithms, and is slightly faster than the best-known deterministic complexity $\Oh(2^{8.038\sqrt n})$ in the unbounded-weight setting by \cite{PBvR17}.

\section{Forest Optimization Problems}

In this section, we examine how we can apply meandric systems to the problems in which we have to find an optimal subgraph of the given graph that is a forest.

\subsection{Forest Meandric Systems}

Similarly to the tree optimization problems, we can define the Eulerian outline $\mathcal E(S)$ of a forest $S$ as the collection of the Eulerian outlines of its subtrees.
The intersection of $\mathcal E(S)$ with a noose $O_e$ still can be encoded by brackets as a non-crossing partition.
However, in the dynamic programming, the restrictions for the double labelings on the frontier are looser because now a submeander can enclose a separate subtree of $S$.

Next, we examine the constraints
imposed on a double bracket labeling $s = (s_L,s_R)$ on the frontier.
Suppose that $m$ is the number of crossings of $\mathcal E(S)$ with the frontier.
Now, the forbidden submeanders are the ones that correspond to a cycle in $S$.
As described in \Cref{sec:cds}, such submeanders start at crossings at even positions and end at crossings at odd positions in $s$.
We call such configurations \emph{forest meandric systems}; see
\Cref{fig:Forest-no-instance} for an example of a forbidden meandric system which is induced by a subgraph with a cycle.

\begin{figure}[ht!]
    \centering
    \includegraphics[page=6,width=0.4\textwidth]{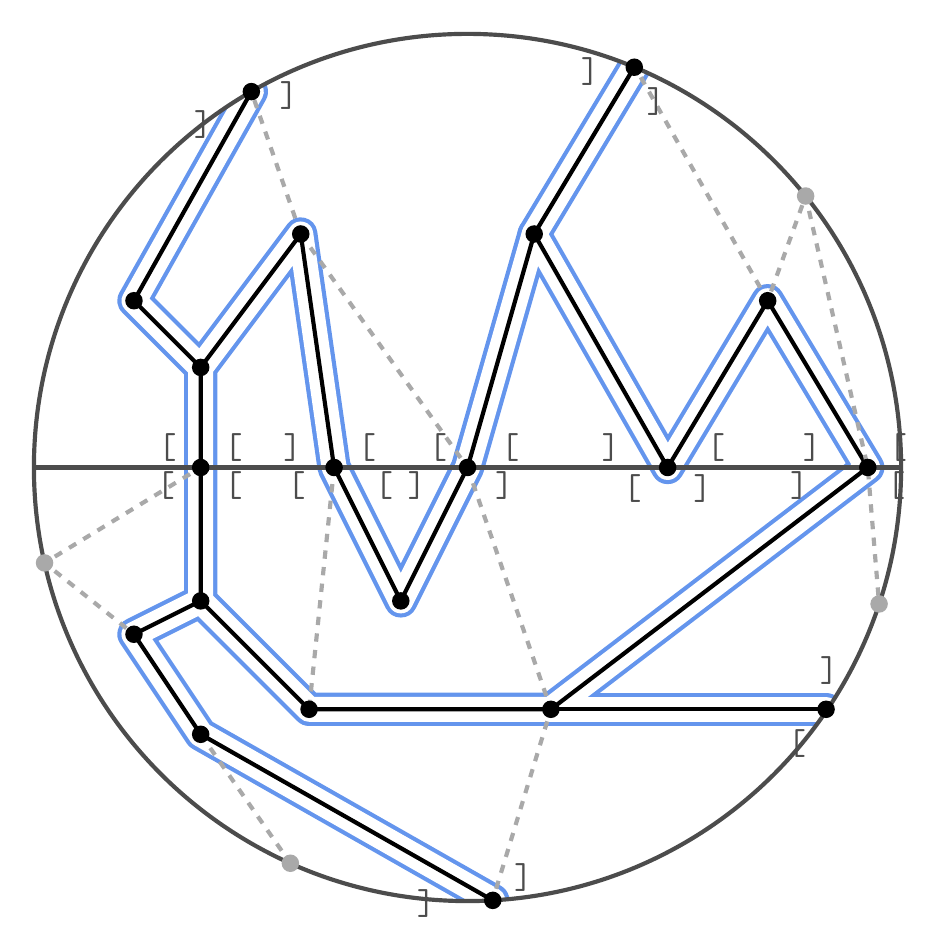}

    \caption{Illustration of a forbidden forest meandric system.}
    \label{fig:Forest-no-instance}
\end{figure}

As far as we know, unfortunately, there is only one previous work that examined similar systems \cite{Bac99}, studying the {\em closed} forest meandric systems.
In our case, we are interested in {\em open} forest meandric systems, for which we have found no calculations or asymptotic bounds.

Let $\Theta^*(\phi^{2m})$ be the asymptotic number of forest meandric systems on $2m$ crossings (we look at even numbers of crossings, since each vertex on the frontier generates two brackets on each side).
Trivially, $\phi \leq 4$, but already by the following simple argument, we can see that it is strictly smaller.
Partition $s$ into consecutive double labeling sequences of length~$4$, and let $s^{(i)} = (s^{(i)}_L,s^{(i)}_R)$ be the $i$-th sequence.
Among the $16^2$ possible labeling options for $s^{(i)}$, we can exclude any double labeling $((*,\texttt{[},\texttt{]},*),(*,\texttt{[},\texttt{]},*))$, where $*$ is any bracket.
The number of such options is $2^4 = 16$.
Thus, we have $\phi \leq \sqrt[4]{16^2-16} < 3.936$.

Similarly to \cite{AP05}, we can use the cluster method to get a better upper bound on $\phi$.
Below, we give a self-contained proof of an improved bound using a simplified version of the cluster method.
A similar argument is also outlined in \cite[Section~5]{AP05}.

\begin{theorem} \label{thm:phi-bound}
    $\phi \leq 3.851$.
\end{theorem}

\begin{proof} \label{phi-bound}
    For a language $\mathcal L \subseteq \Sigma^*$, examine the generating function $G_{\mathcal L}(x)$, where the coefficient at $x^i$ is the number of words of length $i$ in $\mathcal L$.
    If for each $i$ the number of words of length $i$ is $\Oh(c^i)$ for some constant $c$, then the radius of convergence of $G_{\mathcal L}(x)$ is at least $1/c$.
    Thus, the task of determining the asymptotic number of words in $\mathcal L$ of fixed length is reduced to constructing the generating function of $\mathcal L$ and calculating its convergence radius.

    Let $\mathcal W \subset \Sigma^*$ be a set of words that are forbidden as substrings.
    In the following simple scenario, suppose that no two words in $\mathcal W$ overlap (that is, for any two $a, b \in \mathcal W$, $a$ is not contained in $b$ as a substring and no proper suffix of $a$ is a proper prefix of $b$). 
    Examine the generating function
    \begin{equation} \label{eq:gf}
        \sum_{i=0}^{\infty} (G_{\Sigma}(x) - G_{\mathcal W}(x))^i = \frac{1}{1-G_{\Sigma}(x) + G_{\mathcal W}(x)}.
    \end{equation}
    By the principle of inclusion-exclusion, this is equal to $G_{\mathcal L}(x)$ as follows.
    If a word contains $k > 0$ occurrences of forbidden words, then it will contribute $2^k$ times to the series; for a subset of size~$s$ of~$k$ of these words, this will contribute a coefficient $(-1)^s$ to the sum.
    Thus, in total, it will contribute~$(1-1)^k = 0$.
    In the case when the word does not contain forbidden words, it contributes~$1$.
    Hence, the coefficient at $x^i$ correctly counts the number of words in $\mathcal L_{\mathcal W} = \Sigma^* \setminus (\Sigma^* \mathcal W \Sigma^*)$.

    As in \Cref{thm:z-bound}, we will use the alphabet $\{\texttt{U}, \texttt{D}, \texttt{L}, \texttt{R}\}$ for $(\texttt{[},\texttt{[})$, $(\texttt{]},\texttt{]})$, $(\texttt{[},\texttt{]})$ and $(\texttt{]},\texttt{[})$ to encode our meandric systems.
    We will forbid the words $\mathcal W = \texttt{U}\mathcal S \texttt{D}$, where $\mathcal S \subset \{\texttt{L},\texttt{R}\}^*$ is the language of words that contain an equal number of \texttt{L}s and \texttt{R}s.
    Each such word corresponds to a submeander for the following reason.
    Consider the vertical cross-section of the system at a point labeled with \texttt{L} or \texttt{R}; the effect of this label is simply shifting the endpoints of the curves on the cross-section up or down, respectively, creating a single new crossing with the horizontal line.
    The condition of the equal number of \texttt{L} and \texttt{R} ensures that the overall shift is zero.
    Thus, the curves passing through the starting and the ending crossings, labeled by \texttt{U} and \texttt{D}, respectively, are parts of the same submeander; see~\cref{fig:forbidden} for an example.
    Therefore, forbidding the words in $\cal W$ that start at an even position will give the desired upper bound.
    
    \begin{figure}[ht!]
        \centering
        \scriptsize
        \begin{tikzpicture}[scale=.7]
                \draw (-0.5,0) -- (7.5,0);
            
                \foreach \x in {0,1,2,3,4,5,6,7}
                    \fill (\x,0) circle (2pt);
                    
                \draw
                (0,0)   to[out=90,in=90,looseness=1.3] (3,0)
                        to[out=-90,in=-90,looseness=1.3] (6,0)
                        to[out=90,in=90,looseness=1.3] (7,0)
                        to[out=-90,in=-90,looseness=1.3] (2,0)
                        to[out=90,in=90,looseness=1.3] (1,0)
                        to[out=-90,in=-90,looseness=1.3] (0,0);
            
                \draw
                (-0.5,+1.5) to[out=0,in=100,looseness=1.3] (4,0)
                          to[out=-80,in=-100,looseness=1.3] (5,0)
                          to[out=80,in=180,looseness=1.3] (+7.5,+1);
    
                \node[left, yshift=+0.2cm] at (0,0) {\texttt{[}};
                \node[left, yshift=+0.2cm] at (1,0) {\texttt{[}};
                \node[right, yshift=+0.2cm] at (2,0) {\texttt{]}};
                \node[right, yshift=+0.2cm] at (3,0) {\texttt{]}};
                \node[right, yshift=+0.2cm] at (4,0) {\texttt{]}};
                \node[left, yshift=+0.2cm] at (5,0) {\texttt{[}};
                \node[left, yshift=+0.2cm] at (6,0) {\texttt{[}};
                \node[right, yshift=+0.2cm] at (7,0) {\texttt{]}};
    
                \node[left, yshift=-0.2cm] at (0,0) {\texttt{[}};
                \node[right, yshift=-0.2cm] at (1,0) {\texttt{]}};
                \node[left, yshift=-0.2cm] at (2,0) {\texttt{[}};
                \node[left, yshift=-0.2cm] at (3,0) {\texttt{[}};
                \node[left, yshift=-0.2cm] at (4,0) {\texttt{[}};
                \node[right, yshift=-0.2cm] at (5,0) {\texttt{]}};
                \node[right, yshift=-0.2cm] at (6,0) {\texttt{]}};
                \node[right, yshift=-0.2cm] at (7,0) {\texttt{]}};
            \end{tikzpicture}
        \caption{A meandric system corresponding to a forbidden factor $\texttt{ULRRRLLD} \in \mathcal W$.}
        \label{fig:forbidden}
    \end{figure}
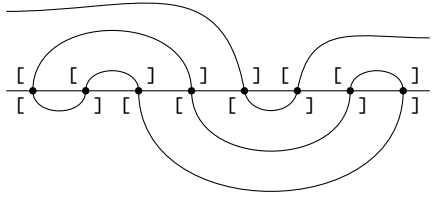

    To deal with the parity of positions, we will construct the generating function using the following $2 \times 2$ matrix:
    \begin{equation}
        A = \begin{pmatrix}
            0 & 4x \\
            4x & 0
        \end{pmatrix} - 
        \begin{pmatrix}
            0 & 0 \\
            0 & G_{\mathcal W}(x)
        \end{pmatrix}.
    \end{equation}
    The idea is that separate generating functions of the strings of even and odd length are encoded in the first and the second coordinates, respectively.
    Adding one of the four coordinates corresponds to multiplying by $4x$ (which changes the length parity); adding a forbidden word multiplies by $G_{\mathcal W}(x)$ (which does not change length parity).
    Similarly to \Cref{eq:gf}, the whole generating function of our language then is given by
    \begin{equation}
        G_{\mathcal L}(x) = 
        \begin{pmatrix}
            1 & 0
        \end{pmatrix}
        \cdot \frac{1}{I - A} \cdot
        \begin{pmatrix}
            1 \\ 0
        \end{pmatrix}.
    \end{equation}
    
    First we determine $G_{\mathcal W}(x)$.
    By the definition of forbidden factors, there are $\binom{2i}{i}$ distinct words of length $2i+2$ in $\mathcal W$.
    Hence,
    \begin{equation}
        G_{\mathcal W}(x) = \sum_{i=0}^{\infty} \binom{2i}{i} x^{2i+2} = \frac{x^2}{\sqrt{1-4x^2}},
    \end{equation}
    where the second equality comes from the central binomial generating function.

    Therefore, we are interested in
    \begin{equation}
        G_{\mathcal L}(x) = 
        \begin{pmatrix}
            1 & 0
        \end{pmatrix} \cdot
        \begin{pmatrix}
            1 & -4x \\
            -4x & 1+\frac{x^2}{\sqrt{1-4x^2}}
        \end{pmatrix}^{-1} \cdot
        \begin{pmatrix}
            1 \\ 0
        \end{pmatrix}.
    \end{equation}
    Calculating this gives
    \begin{equation}
        G_{\mathcal L}(x) = \frac{1+\frac{x^2}{\sqrt{1-4x^2}}}{1-16x^2+\frac{x^2}{\sqrt{1-4x^2}}}.
    \end{equation}
    The convergence radius of this series is $r \approx 0.259676$, hence $\phi \leq 1/r \approx 3.851$.
\end{proof}

We can obtain a better upper bound by using the Goulden-Jackson cluster method, similarly to how we did for the meandric numbers.
For that, we adapt the language of forbidden words to include only submeanders that start at even positions and end at odd positions.
Using the program provided by ChatGPT-5.6 Sol, with the length of standard representatives $n = 24$, we obtain an improved upper bound:
\begin{theorem}
    $\phi \leq 3.801179$.
\end{theorem}
The full code and results are available at \cite{BLCV26}.

\subsection{Planar Feedback Vertex Set}

We apply the meandric systems to the \textsc{Planar Feedback Vertex Set} problem, which can be equivalently formulated as the problem of finding a subset of vertices $S$ with the largest total weight such that $S$ induces a forest in $G$.

We use a similar labeling as in the tree problems with $5$ options in $\{\texttt{0}, (\texttt{[},\texttt{[}), (\texttt{[},\texttt{]}), (\texttt{]},\texttt{[}), (\texttt{]},\texttt{]})\}$ for each vertex.
The label \texttt{0} means that the vertex is not in $S$, and the bracket pairs denote the intersection topology of the Eulerian outline with the noose.
We can then analyze the relevant parameters $c_1$ and $c_2$:
\begin{itemize}
    \item Since there are $5$ possible labels for each vertex, $c_1 = 5$.
    \item For the frontier, we can bruteforce the set $A$ to include in $S$.
    Each vertex in $F \setminus A$ is assigned the label \texttt{0}.
    On the set $A$, the number of options is upper bounded by $\Oh^*(\phi^{2|A|})$, since we are forbidding cycles in $S$.
    Similarly as before, the total number of options is
    \begin{equation}
        \mathcal F = \Oh^*\Bigl(\sum_{A \subseteq F} \phi^{2|A|}\Bigr) =  \Oh^*\left((1+\phi^2)^{|F|}\right).
    \end{equation}
    Therefore, $c_2 = 1+\phi^2$.
\end{itemize}

By applying \Cref{thm:general-dp} and \Cref{thm:general-fmm}, we obtain the following theorem:
\begin{theorem} \label{thm:fvs}
    There exist deterministic algorithms that solve the \textup{\textsc{Planar Feedback Vertex Set}} problem in $\Oh^*((5 \sqrt{1+\phi^2})^{\bw}) = \Oh(2^{9.119\sqrt{n}})$ and $\Oh^*(W\cdot (1+\phi^2)^{\bw}) = \Oh(W\cdot 2^{8.382\sqrt{n}})$ time.
\end{theorem}

While these results do not improve the best-known deterministic complexity in the unbounded-weight setting $\Oh^*(2^{8.038\sqrt{n}})$ from \cite{PBvR17}, they give an illustration of how meandric systems can be used for the forest connectivity problems.

\section{Conclusions and Open Problems}
In this research, we revisited exact algorithms for planar connectivity problems based on sphere cut decomposition from the perspective of meandric systems. This allowed us to identify and exclude frontier configurations that cannot be extended to a feasible global solution. By combining this idea with the fast-matrix multiplication of Dorn~\cite{Dor06}, we obtained improved deterministic bounds in the bounded-weight setting for several classical planar problems, namely \textsc{Planar TSP}, \textsc{Planar Longest Cycle/Path}, \textsc{Planar Connected Dominating Set} and \textsc{Planar Steiner Tree}. Furthermore, we extended the meandric interpretation beyond cycle problems, using Eulerian outlines, and applied related counting arguments to tree and forest optimization problems, such as \textsc{Planar Steiner Tree}, \textsc{Planar Connected Dominating Set}, and \textsc{Planar Feedback Vertex Set}.

A central question left open by our work is to better understand the asymptotic growth of the meandric systems that arise in these dynamic programming algorithms. 
Our running times depend directly on upper bounds for the growth constants of acyclic unidirectional open meandric systems and the forest open meandric systems introduced here. Any improvement on these constants would immediately translate into improved algorithmic bounds.
%

\section{Acknowledgements}

The ideas and proofs presented in this paper were developed by the authors without the use of LLMs.
We have used ChatGPT-5.6 Sol to implement the Goulden-Jackson cluster method for the meandric system upper bounds, with a few conceptual corrections to the code.
We then checked the correctness of the code manually.
We have also used LLMs to check the correctness of the manuscript and to validate the corrections of \Cref{app:mistakes}.

J.V.~is supported by the 1.1.1.9 Research application No 1.1.1.9/LZP/2/25/206 of the Activity ``Post-doctoral Research'' ``Practical applications and limitations of quantum search''.

\printbibliography

\newpage

\begin{appendices}

\crefalias{section}{appendix}

\section{Reductions to Bridgeless Instances} \label{app:problems}
In this appendix, we show that, for each of the aforementioned problems, assuming that the input graph is bridgeless entails no loss of generality. 
In particular, we describe how to transform the input instance into an equivalent bridgeless instance while preserving both the existence of a solution and its optimality.

\paragraph{Planar Travelling Salesman Problem.} The case of the {\sc Planar Travelling Salesman} problem is the simplest to address. In fact, a bridge does not belong to any cycle. Hence, a graph containing a bridge is a trivially infeasible instance.

\paragraph{Planar Longest Cycle Problem.} Since bridges do not belong to any cycle, the graph $G'=(V,E \setminus B)$, where $B$ denotes the set of all bridges in $G$, contains exactly the same set of cycles as $G$.
Consequently, computing the longest cycle in $G'$ is equivalent to computing a longest cycle in $G$.

\paragraph{Planar Longest Path Problem.} 
Let $(G,w)$ be an instance of the {\sc Planar Longest Path} problem.
Observe that a bridge $e=\{u,v\}$ may belong to a longest path $P$ of $G$.
We perform the following reduction. For each bridge $e=\{u,v\}$ in $G$, we introduce a new vertex $z_e$ and add two new edges $\{u,z_e\}$ and $\{z_e,v\}$, each assigned weight $0$. 
After applying this construction, the resulting graph $G'$ is bridgeless and planar. Also, let $w': E(G') \rightarrow \mathbb{N}_0$ denote the resulting weight function for the edges of $G'$. 
We next show that $(G,w)$ admits a longest path $P$ of weight $k$ if and only if $(G',w')$ admits a longest path $P'$ of weight $k$.
Let $P$ be a longest path in $(G,w)$. Since $G$ is a subgraph of $G'$, $P$ is a path in $(G',w')$. Also, since all added edges have weight $0$, $P$ is a longest path in $(G',w')$.
For the opposite direction, let $P'$ be a longest path in $(G',w')$, we show that there exists a path in $(G',w')$ of the same weight as $P'$ that does not contain any vertex $z_e$ introduced in the reduction. Clearly, such a path is entirely contained in $G$.
We argue this in two steps.
First, suppose that $P'$ \emph{passes through} some added vertex $z_e$, i.e., the edges $\{u,z_e\}$ and $\{z_e,v\}$ belong to $P'$. We consider two cases: $e$ has weight $0$ or has weight larger than $0$. In the first case, since the sub-path $(u,z_e,v)$ contributes $0$ to the total weight, replacing it with the edge $\{u,v\}$, yields a path of equal length.
The second case is not possible since the same replacement would yield a path of length greater to the one of $P'$.
Second, suppose a longest path $P'$ in $G'$ uses at least a vertex $z_e$ as an endpoint. Then $P'$ begins with one of the two zero-weight edges incident to $z_e$, say $\{z_e,u\}$, and continues from $u$. Hence, the sub-path of $P'$ starting from $u$ is itself a longest path in $G'$. Thus, removing $z_e$ yields a path of the same length. 

\paragraph{Planar Steiner Tree Problem.} 
Let $(G,w)$ be an instance of the {\sc Planar Steiner Tree} problem.
In order to obtain an equivalent bridgeless instance, we perform the following reduction. For each bridge $e=\{u,v\}$ in $G$, we introduce a new vertex $z_e$ and add two edges $\{u,z_e\}$ and $\{z_e,v\}$, each assigned weight $w(e)$. 
After applying this construction, the resulting graph $G'$ is bridgeless and planar. Also, let $w': E(G') \rightarrow \mathbb{N}$ denote the resulting weight function for the edges of $G'$. 
We show next that $(G,w)$ admits Steiner tree $T$ of weight $k$ if and only if $(G',w')$ admits a Steiner tree $T'$ of weight $k$. 
The proof is analogous to the one for the {\sc Planar Longest Path} problem. Let $T$ be a minimum-weight Steiner tree in $(G,w)$. Note that $T$ is also a Steiner tree in $(G',w')$ since $G \subset G'$ and since no vertex in $z_e$ is a terminal vertex.
For the opposite direction, by the same arguments as for the {\sc Planar Longest Path} problem, one can easily see that there exists a Steiner tree in $(G',w')$ that uses no edge incident to an added vertex $z_e$. Clearly, such a Steiner tree belongs to $G$ as well and has the same weight as $T'$.

\paragraph{Planar Connected Dominating Set Problem.} 
Let $(G,w)$ be an instance of the {\sc Planar Connected Dominating Set} problem. Recall now that $w: V(G) \rightarrow \mathbb{N}$.
W.l.o.g., assume that $G$ does not consist of a single edge.
In order to obtain an equivalent bridgeless instance, we perform the following reduction. For each bridge $e=\{u,v\}$ in $G$, we introduce a new vertex $z_e$, with $w(z_e) = 1$, and add two edges $\{u,z_e\}$ and $\{z_e,v)\}$.
After applying this construction, the resulting graph $G'$ is bridgeless. Also, let $w': V(G') \rightarrow \mathbb{N}$ denote the resulting weight function for the vertices of $G'$. 
We argue next that $(G,w)$ admits a connected dominating set $D$ of weight $k$ if and only if $(G',w')$ admits a  connected dominating set $D'$ of weight $k$. 
Let $D$ be a minimum-weight connected dominating set in $(G,w)$. We argue that $D$ is also a connected dominating set in $(G',w')$. First, note that all the vertices in $V(G') \setminus \{z_e: e \text{ is a bridge in } G\}$ are dominated by $D$, since $G \subset G'$. Second, we show that, for each bridge $e$ of $G$, at least one of the two neighbors $u$ and $v$ of $z_e$ belongs to $D$; which implies that each vertex $z_e$ is also dominated by $D$. Suppose, for a contradiction, that neither $u$ nor $v$ belongs to $D$. Note that, removing the edge $e$ from $G$, results in two connected components $C_u$ and $C_v$ such that $u$ belongs to $C_u$ and $v$ belongs to $C_v$. Then $u$ must be dominated by some vertex in $V(C_u) \cap D$ and $v$ must be dominated by some vertex in $V(C_v) \cap D$. However, these vertices are not in the same connected component in the subgraph induced by $D$, a contradiction. 
Again, for the opposite direction, with similar arguments as for the {\sc Planar Longest Path} and the {\sc Planar Steiner Tree}  problem, which we elaborate hereafter, one can easily see that there exists a minimum-weight connected dominating set $D'$ in $(G',w')$ that uses none of the added vertices $z_e$. Let $D'$ be a connected dominating set of $(G',w')$ of minimum weight $k$, and suppose, for a contradiction, that $D'$ contains some vertex $z_e$, where $e = \{u,v\}$. Since $z_e$'s only neighbors in $G'$ are $u$ and $v$, the connectivity of $D'$ forces at least one of $u$ and $v$ to belong to $D'$. As $(u,v)$ is an edge of $G'$, the vertices dominated by $z_e$ (namely $u$ and $v$) are also dominated by whichever of $u$ and $v$ lies in $D'$, and the edge $\{u,v\}$ preserves the connectivity among the remaining vertices of $D'$ once $z_e$ is removed. Hence $D' \setminus \{z_e\}$ is still a connected dominating set of $(G',w')$ with smaller weight; a contradiction.

\paragraph{Planar Feedback Vertex Set Problem.} In the {\sc Planar Feedback Vertex Set} problem, since a bridge belongs to no cycle of $G$ by definition, to reduce to a bridgeless instance, one simply deletes all bridges from $G$. The resulting graph $G'$ has the same set of cycles as $G$, and, consequently, has the same minimum feedback vertex set.

\section{Complexity of the Planar Longest Cycle/Path} \label{app:longest}

In this appendix, we assume close familiarity with \cite{BCKN15,PBvR17} and calculate the complexity of the \textsc{Planar Longest Cycle/Path} in the latter work.

First, we note that the randomized algorithm for the \textsc{Planar Travelling Salesman} problem from \cite{PBvR17} can be adapted with minimal changes to the \textsc{Planar Longest Cycle} problem.

\begin{theorem}
    There exists a Monte Carlo algorithm that solves the \textup{\textsc{Planar Longest Cycle/Path}} in $\Oh^*((4^{\omega/2})^{\bw}) = \Oh(2^{5.034\sqrt{n}})$ time.
\end{theorem}

\begin{proof}[Proof sketch.]
This theorem follows from the analysis of Theorem 4.11 in the Master's thesis \cite{Pin16}, which describes the TSP algorithm.
When calculating the table $A_x(i_x,w_x,s_x)$, we also need to consider compatible labeling $s_y(v) = \texttt{0}$ and $s_z(v) = \texttt{0}$ for a vertex $v \in F$, which corresponds to $v$ not being in the cycle.
However, this doesn't change the asymptotic complexity, since that is governed by the number of labels used for each vertex, which is still $4$.

The \textsc{Planar Longest Path} problem can be reduced to the \textsc{Planar Longest Cycle}, only with a constant increase in the branchwidth of the graph, see \cite[Appendix A.4]{CNPPVRW22}.
\end{proof}

However, we observe that for the deterministic algorithm, there is an increase in the complexity:

\begin{theorem}
    There exists a deterministic algorithm that solves the \textup{\textsc{Planar Longest Cycle/Path}} in $\Oh^*(((2+2^{\omega/2})\cdot \sqrt{5})^{\bw}) = \Oh(2^{6.912\sqrt{n}})$ time.
\end{theorem}

\begin{proof}[Proof sketch.]
We can adapt the \textsc{Planar Travelling Salesman} algorithm from \cite{PBvR17} with small changes.
The corresponding algorithm is described in \cite{Pin16}.

Again, the most substantial change is allowing pairs of labels $(\texttt{0},\texttt{0})$ on $F$.
For a vertex $f \in F$ and children bags $y$ and $z$ of $x$, now we consider labels $s_y(f)$ and $s_z(f)$ compatible if $s_y(f) + s_z(f) \in \{0,2\}$ (compared to $s_y(f) + s_z(f) = 2$ in the TSP case).

This has an effect on page 43 in \cite{Pin16} in the expression $\sum_{s_y(f) + s_z(f) = 2} 2^{[s_y(f)=1]/2} 2^{[s_y(f)=1]/2}$, where $f \in F$ is a vertex in the forget set, which now becomes
\[\sum_{s_y(f) + s_z(f) \in \{0,2\}} 2^{[s_y(f)=1]/2} 2^{[s_z(f)=1]/2}.\]
Since now the compatible pairs of labels are $(\texttt{1},\texttt{1}), (\texttt{2},\texttt{0}), (\texttt{0},\texttt{2}), (\texttt{0},\texttt{0})$, this sum is equal to $5$ (compared to $4$ in the TSP case).

This changes the complexity to
\[\Oh^*((2+2^{\omega/2})^{|L|+|R|} 5^{|F|}) = \Oh^*(((2+2^{\omega/2})\cdot \sqrt{5})^{\bw}),\]
as $|L|+|R|=\bw$ and $|F|=\bw/2$ in the worst case.

Finally, by the arguments of \cite[Section 5]{DPBF10}, we can adapt the DP on the sphere cut decomposition to the \textsc{Planar Longest Path} without increasing the asymptotic complexity.
\end{proof}

\section[Proof of Theorem~\ref*{thm:general-fmm}]
        {Proof of \Cref{thm:general-fmm}} \label{app:fmm}

In this appendix, we describe and analyze the fast matrix multiplication method by~\cite{Dor06} for completeness.

First, we provide the background of fast matrix multiplication.
Two $n \times n$ matrices can be multiplied in $\Oh(n^\omega)$ time, where $\omega$ is the fast matrix multiplication constant.
The best current upper bound on the constant is $\omega < 2.371339$ \cite{ADWXXZ25}.
For $k \geq n$, to multiply two rectangular matrices with sizes $n \times k$ and $k \times n$, we first partition them into $\Oh(k/n)$ blocks of size $n \times n$, padding if necessary.
Then we multiply the corresponding blocks using fast matrix multiplication and sum up all the results.
This approach gives complexity $\Oh((k/n)\cdot n^{\omega}) = \Oh(kn^{\omega-1})$.

Suppose now that we need to multiply rectangular matrices with dimensions $n \times k$ and $k \times m$, where $k \geq n$.
We can assume that $m \geq n$, as the other case follows in the same way.
We then can partition the second matrix into $\Oh(m/n)$ matrices of size $k \times n$
and apply the above algorithm for each of them, afterwards concatenating the results.
This gives complexity $\Oh((m/n)\cdot k n^{\omega-1}) = \Oh(k m n^{\omega-2})$.
If $k < n$ and $m \geq n$, the same approach results in complexity $\Oh(m n^{\omega-1})$.

A related notion is the matrix \emph{distance product}, also known as the \emph{min-plus} matrix multiplication.
For two matrices $A$ and $B$ with sizes $n \times k$ and $k \times m$, it is defined as an $n \times m$ matrix $A\star B$, where
\begin{equation}
    (A\star B)_{ij} = \min_{s = 1}^k \{ A_{is} + B_{sj} \}.
\end{equation}
If the matrix entry values do not exceed $W$, then fast matrix multiplication can be applied to calculate the distance product of two $n \times n$ matrices in $\Oh(W\cdot n^{\omega})$ time~\cite{Zwi02}.
For our case of rectangular matrix multiplication, it implies $\Oh(W\cdot m k n^{\omega-2})$ complexity for $n \leq k$ and $\Oh(W\cdot m n^{\omega-1})$ complexity for $n > k$.

\begin{proof}
    We illustrate the algorithm on the \textsc{Planar TSP}, but the general principle is the same for other problems.
    We can speed up the merging step of the DP tables of $e_L$ and $e_R$ in the dynamic programming algorithm.
    Let $\mathcal B_L$ and $\mathcal B_R$ be the total number of possible labelings to the vertex set $L \cup I$ and $R \cup I$, respectively (essentially, the left and the right side of the boundary).
    For TSP, we have $\mathcal B_L = \Oh(4^{|L|})$ and $\mathcal B_R = \Oh(4^{|R|})$, as $|I| \leq 2$.
    
    For $e_L$, we construct a $\mathcal B_L \times \mathcal F$ matrix $A$ as follows.
    Let $i$ be a labeling for $L \cup I$ and let $s = (s_L,s_R)$ be a valid pair of labelings of the frontier, meaning it can potentially lead to a feasible solution, where $s_L$ and $s_R$ are the labelings on the left and the right sides of the frontier, respectively.
    We also require $i \cup s_L$ to be a valid labeling of $O_L$.
    Checking the validity in both cases requires time polynomial in $\ell$ as discussed in \Cref{sec:dp}.
    The value of $A_{is}$ then is equal to the minimum total weight of a partial Hamiltonian cycle in $G_L$ that agrees with the labeling $i \cup s_L$, which is equal to $D[e_L][i \cup s_L]$ (which is $+\infty$ if $s$ is not valid).
    Constructing $A$ thus takes $\Oh^*(\mathcal B_L \mathcal F)$ time.

    Similarly, for $e_R$, we construct an $\mathcal F \times \mathcal B_R$ matrix $B$.
    Then for the distance product $C = A \star B$, which is a matrix of size $\mathcal B_L \times \mathcal B_R$, the entry $C_{ij}$ denotes the weight of the shortest partial Hamiltonian cycle that agrees with the labeling $i \cup j$ on the boundary, or $+\infty$ if it does not exist.
    This is because for each labeling $s$ on the frontier, we constructed $A$ and $B$ by considering only valid labeling pairings $i \cup s$ for $e_L$ and $s \cup j$ for $e_R$.
    From the matrix $C$ then we can obtain $D[e_P]$ in $\Oh^*(\mathcal B_L \mathcal B_R) = \Oh^*(\mathcal B)$ time.

    Now we analyze the complexity; without loss of generality, assume that $\mathcal B_L \leq \mathcal B_R$.
    \begin{itemize}
    
        \item If $\mathcal B_L \leq \mathcal F$, the complexity of the matrix multiplication is
        \begin{equation}
            \Oh^*(W\cdot \mathcal B_L^{\omega-2} \mathcal B_R \mathcal F) = \Oh^*(W\cdot c_1^{|L|(\omega-2)} c_1^{|R|} c_2^{|F|}).
        \end{equation}
        We can again use $|L|+|F|\leq \ell$ and $|R|+|F|\leq \ell$.
        If $|F| > \ell/2$, then reducing $|F|$ by $1$ can increment both $|L|$ and $|R|$ by $1$.
        Since $c_1 = 4$ and $c_2 = 6$, we have $c_1^{\omega-2} c_1 = c_1^{\omega-1} \approx 6.7 > c_2$ when $\omega = 2.372$.
        Therefore, in the worst case, $|L|+|R| = \ell$ and $|F| \leq \ell/2$.

        Since $\mathcal B_L \leq \mathcal B_R$, we have $|L| \leq |R|$.
        Then in the worst case $|L| = |F| = \ell - |R|$.
        The total complexity is then 
        \begin{equation}
            \Oh^*(W \cdot c_1^{(\ell-|R|)(\omega-2)} c_1^{|R|} c_2^{\ell-|R|}) = \Oh^*((c_1^{\omega-2}c_2)^{\ell-|R|} c_1^{|R|}).
        \end{equation}
        As $c_1^{\omega-2} c_2 > c_1$, this is the largest when $|R|$ is minimal, so in the worst case $|L| = |R| = |F| = \ell/2$, and the overall complexity is
        \begin{equation}
            \Oh^*(W\cdot c_1^{(\ell/2)\cdot (\omega-1)} c_2^{\ell/2}) = \Oh^*(W\cdot (4^\frac{\omega-1}{2} \sqrt{6})^{\bw}).
        \end{equation}

        If instead it happened that $c_1^{\omega-1} < c_2$, then the worst case would occur when $|F| = \ell$ and this part would take time $\Oh^*(c_2^{\ell})$.
    
        \item In the other case, we have $\mathcal B_L > \mathcal F$.
        The complexity of the matrix multiplication then is
        \begin{equation}
            \Oh^*(W\cdot \mathcal B_L^{\omega-1} \mathcal B_R) = \Oh^*(W\cdot c_1^{|L|(\omega-1)} c_1^{|R|}).
        \end{equation}
        Once again, $|L|$ has to be maximized for the worst case, so then $|L| = |R| = \ell/2$.
        This results in complexity
        \begin{equation}
            \Oh^*(W\cdot c_1^{(\ell/2)\cdot (\omega-1)} c_1^{\ell/2}).
        \end{equation}
        This is dominated by the complexity of the previous case as $c_2 \geq c_1$.
        
    \end{itemize}

    In the case of the \textsc{Planar TSP}, the total complexity becomes $\Oh^*(W\cdot \max(4^\frac{\omega-1}{2} \sqrt{6}, 6)^{\bw}) = \Oh^*(W\cdot (4^\frac{\omega-1}{2} \sqrt{6})^{\bw})$.
\end{proof}

\section{Corrections to Previous Works} \label{app:mistakes}

In this appendix, we discuss several inaccuracies and inconsistencies that we identified in the relevant literature. Clarifying these issues is important for a fair comparison between our algorithms and previous results, which in some cases have been incorrectly stated.
We use the $\leadsto$ arrow to denote a correction.

Recapping, \cite{DPBF10} introduces dynamic programming on the sphere cut decomposition; \cite{Dor06} speeds these DP algorithms using fast matrix multiplication; \cite{Dor07} incorporates the results from the first two papers.
The work \cite{PBvR17} gives state-of-the-art deterministic and randomized complexities for connectivity problems in the unbounded-weight setting. 

\paragraph{Planar Travelling Salesman Problem.}

In \cite{DPBF10,Dor06,Dor07} the problem termed the \textsc{Travelling Salesman} problem is what is usually called the \textsc{Travelling Salesman Walk} problem \cite{LN08}.
In the first problem, as it is usually defined, we have to find the shortest cycle that visits each city \emph{exactly} once (a Hamiltonian cycle).
In the second problem, we need to visit each city \emph{at least} once; this corresponds to a closed \emph{walk} in a graph.

Most probably this has led to a confusion in Table 4 in \cite{PBvR17}, where the complexities of Dorn, Penninkx, Bodlaender and Fomin for \textsc{Planar Hamilton Cycle/TSP} are actually the complexities of the \textsc{Planar Travelling Salesman Walk} problem of \cite{DPBF10,Dor06,Dor07}, and should be
\begin{equation}
    \Oh(2^{9.86\sqrt{n}})
    \quad\leadsto\quad
    \Oh(2^{6.903\sqrt{n}})
\end{equation}
and
\begin{equation}
    \Oh(2^{8.15\sqrt{n}})
    \quad\leadsto\quad
    \Oh(2^{5.579\sqrt{n}}).
\end{equation}

Table 4 also states complexity $\Oh(2^{5.63\sqrt{n}})$ for the \textsc{Planar Travelling Salesman} problem, which seems to be a typo (the correct complexity is $\Oh(2^{6.570\sqrt n})$, and the algorithm is explained in \cite{Pin16}).

\paragraph{Planar Longest Cycle/Path.}

For \textsc{Planar Longest Cycle/Path}, \cite{DPBF10} gives the dynamic programming algorithm with complexity $\Oh(2^{7.223\sqrt{n}})$ (Theorem 5), which is consistent with $\Oh^*((4\sqrt{7})^{\bw})$, see \Cref{thm:plc-meanders}.

Then \cite{Dor07} uses a bound on acyclic meandric systems $z \leq 3.68133$ to obtain an $\Oh(2^{7.152\sqrt{n}})$ algorithm, which is consistent with $\Oh^*((4\sqrt{3+z})^{\bw})$.
Then they apply fast matrix multiplication to obtain complexity $\Oh(2^{5.828\sqrt{n}})$, which is consistent with $\Oh^*((4^{\frac{\omega-1}{2}}\sqrt{3+z})^{\bw})$.
The analysis of the last complexity includes a mistake.
The first equation in Lemma 3.2.6.~should say: 
\begin{equation}
    Z := \left(\frac{z}{2}+1\right)^{\frac{\ell}{2}}\cdot 3^{\frac{\ell}{2}}
    \quad\leadsto\quad
    Z := \left(\frac{z}{3}+1\right)^{\frac{\ell}{2}}\cdot 3^{\frac{\ell}{2}}.
\end{equation}
The corrected expression clearly is equal to $(3+z)^{\frac{\ell}{2}}$, which is consistent with \Cref{thm:plc-meanders}.

Finally, \cite{Dor06} states incorrect complexity for \textsc{Planar Longest Cycle/Path}.
The complexity of the DP algorithm of \cite{DPBF10} is stated to be $\Oh(2^{6.903\sqrt n})$, which is incorrect (possibly being taken from the \textsc{Planar Hamiltonian Cycle} problem),
\begin{equation}
    \Oh(2^{6.903\sqrt{n}})
    \quad\leadsto\quad
    \Oh(2^{7.223\sqrt{n}}).
\end{equation}
The complexity of the DP algorithm with fast matrix multiplication seems to suffer a typo,
\begin{equation}
    \Oh(2^{5.58\sqrt{n}})
    \quad\leadsto\quad
    \Oh(2^{5.828\sqrt{n}}).
\end{equation}

\paragraph{Planar Connected Dominating Set.}

Recall the labeling scheme for \textsc{Planar Connected Dominating Set} in \Cref{thm:cds}.
It is not written explicitly in \cite{DPBF10}, but the $\sqrt{2+z^2}$ part in the paper is bounded by $\sqrt{2+15}$.
Essentially, the number $15$ comes from forbidding the configuration $((\texttt{[},\texttt{]}),(\texttt{[},\texttt{]}))$ which leads to a disconnected cycle, out of $16$ possibilities for each vertex on the frontier.
Hence, $c_2 = 17$ in their analysis.
This correctly gives the complexity $\Oh^*((6\sqrt{17})^{\bw}) = \Oh(2^{9.823\sqrt{n}})$.

In \cite{Dor06}, it is stated that fast matrix multiplication improves the complexity to $\Oh(W\cdot 2^{8.11\sqrt{n}})$.
While $\Oh^*((6^{\frac{\omega-1}{2}}\sqrt{17})^{\bw})$ indeed is $\Oh(2^{8.111\sqrt{n}})$, their main Theorem 4 states that the complexity should be $\Oh^*(W\cdot \max( c_1^{(\omega-1)\cdot \frac{\bw}{2}} c_2^{\frac{\bw}{2}}, c_2^{\bw} ))$, consistent with \Cref{thm:general-fmm}.
We have $\Oh^*(c_2^{\bw}) = \Oh(2^{8.674\sqrt{n}})$, so the total complexity should be $\Oh(W\cdot 2^{8.674\sqrt{n}})$ instead,
\begin{equation}
    \Oh(2^{8.11\sqrt{n}})
    \quad\leadsto\quad
    \Oh(2^{8.674\sqrt{n}}).
\end{equation}
The work \cite{PBvR17} also references the value $8.11$.

\paragraph{Planar Steiner Tree.}

For the \textsc{Planar Steiner Tree} problem, it seems that the $4$-symbol labeling in \cite{DPBF10} is incomplete.
The authors use an alphabet of four symbols $\{\texttt{0}, \texttt{[}, \mathtt{\square}, \texttt{]}\}$.
The symbol \texttt{0} encodes either an isolated vertex or a component with the only intersection with the noose being this vertex.
The other symbols encode how a component touches multiple vertices on the noose; the symbol \texttt{[} marks the beginning, $\mathtt{\square}$ is an intermediate vertex and \texttt{]} is the final vertex of a component.

However, using the same symbol \texttt{0} for vertices isolated from the Steiner tree $S$ and being in a component that touches the noose only at this vertex loses some necessary information.
The problem arises, for example, when there are two vertices on the frontier vertex set, where each is assigned the label \texttt{0} on both sides:
\begin{itemize}
    \item It may be the case that both of these two vertices are isolated in the respective partial solution on both sides.
    Note that this situation can still lead to a feasible solution, since both of these vertices are not in the Steiner tree.
    \item It may also be the case that for each of these vertices, there is a connected component on some side that contains a vertex from $X$.
    Observe that this situation cannot lead to a feasible solution, since we are getting two disconnected components that both contain terminal vertices.
\end{itemize}
However, these two cases are identified by the same labels (a pair $(\texttt{0},\texttt{0})$ for each vertex), so it seems to that the described labeling is not sufficient to distinguish between the two. 

In our encoding in \Cref{thm:cds}, we use different labels for these cases (\texttt{0} and $(\texttt{[},\texttt{]})$, respectively).
The symbols $\texttt{[}, \mathtt{\square}, \texttt{]}$ in their encoding clearly match to the pairs $(\texttt{[},\texttt{[})$, $(\texttt{]},\texttt{[})$ and $(\texttt{]},\texttt{]})$, respectively, in our encoding.
Their claimed complexity $\Oh(2^{8.488\sqrt{n}})$ corresponds to $\Oh^*((4\sqrt{16})^{\bw})$ that comes from using $4$ labels.
By using $5$ labels instead of $4$, the complexity becomes $\Oh^*((5\sqrt{1+16})^{\bw}) = \Oh(2^{9.264\sqrt{n}})$,
\begin{equation}
    \Oh(2^{8.488\sqrt{n}})
    \quad\leadsto\quad
    \Oh(2^{9.264\sqrt{n}}).
\end{equation}

The same discrepancy has an effect in \cite{Dor06,Dor07}.
Using fast matrix multiplication, they claim complexity $\Oh(W\cdot 2^{7.16\sqrt{n}})$.
This corresponds to using $4$ labels, $\Oh^*((4^{\frac{\omega-1}{2}} \sqrt{16})^{\bw})$.
Using $5$ labels, the complexity becomes $\Oh^*((5^{\frac{\omega-1}{2}} \sqrt{1+16})^{\bw}) = \Oh(W\cdot 2^{7.717\sqrt{n}})$.
Moreover, similarly to \textsc{Planar Connected Dominating Set}, they omit taking the maximum in Theorem 4 in \cite{Dor06}.
Since here $c_2 = 17$, the complexity is dominated by $\Oh^*(W\cdot c_2^{\bw}) = \Oh(W\cdot 2^{8.674\sqrt{n}})$, so
\begin{equation}
    \Oh(2^{7.16\sqrt{n}})
    \quad\leadsto\quad
    \Oh(2^{8.674\sqrt{n}}).
\end{equation}

The work \cite{PBvR17} also references the faulty values, which means that their deterministic $\Oh(2^{8.039\sqrt{n}})$ algorithm actually beats the fast matrix multiplication algorithm of \cite{Dor06}.

\paragraph{Planar Feedback Vertex Set.}

Firstly, we were unable to reproduce the complexity $\Oh(W\cdot 2^{7.56\sqrt{n}})$ from \cite{Dor06,Dor07}.
The dynamic programming algorithm from \cite{DPBF10} without any meander speedup results in complexity $\Oh((5\sqrt{1+4^2})^{\bw}) = \Oh(2^{9.264\sqrt{n}})$; this corresponds to the same 5 labels we use \cref{thm:fvs}.
Similarly to other fast matrix multiplication results in \cite{Dor07}, we expect the complexity $\Oh(W\cdot (5^{\frac{\omega-1}{2}}\sqrt{1+4^2})^{\bw}) = \Oh(W\cdot 2^{7.717\sqrt{n}})$, which is larger than the claimed complexity.
Unfortunately, references \cite{Dor06,Dor07} do not explain how to obtain this complexity bound.

Secondly, it seems that the complexity for this problem was affected by the same mistake as for \textsc{the Planar Connected Dominating Set} and \textsc{Planar Steiner Tree}.
Since we have $c_2 = 17$, the complexity of the algorithm is dominated by $\Oh(W\cdot c_2^{\bw}) = \Oh(W\cdot 2^{8.674\sqrt{n}})$.
\begin{equation}
    \Oh(2^{7.56\sqrt{n}})
    \quad\leadsto\quad
    \Oh(2^{8.674\sqrt{n}}).
\end{equation}

\end{appendices}

\end{document}

%% file: fig-nooses.tex
\def\Barrow{1.9}
\coordinate (anc) at ({425:\Barrow});
\draw[ForestGreen,-{Latex[length=2mm]}, line width=0.9pt]
  (anc) arc[start angle=425, end angle=80, radius=\Barrow];
  
\small
\def\n{10}
\def\r{2}

\foreach \i in {1,...,\n} {
\ifnum\i=3\else
    \node[circle, fill=black, inner sep=2pt] (v\i)
      at ({360/\n * (\i-1)}:\r) {};
\fi
}

\coordinate (v3) at ({360/\n * (3-1)}:\r);

\foreach \i in {1,...,\n} {
\pgfmathtruncatemacro{\j}{mod(\i,\n)+1}
\draw[bend left=-14] (v\i) to (v\j);
}

\path (v3) -- (v9)
node[pos=0.25, circle, fill=black, inner sep=2pt] (c1) {}
node[pos=0.50, circle, fill=black, inner sep=2pt] (c2) {}
node[pos=0.75, circle, fill=black, inner sep=2pt] (c3) {};

\node[xshift=-.6cm] at (v6) {$L$};
\node[xshift=+.6cm] at (v1) {$R$};
\node[yshift=+.5cm,xshift=+.2cm] at (c1) {$F$};
\node[yshift=-.2cm,xshift=+.5cm] at (v9) {$I_2$};
\node[yshift=+.2cm,xshift=+.5cm] at (v3) {$I_1$};

\large
\node[ForestGreen] at (0,2.3) {$O_P$};
\node[RoyalBlue] at (-.5,0) {$O_L$};
\node[RoyalBlue] at (1.3,0) {$O_R$};
\node[white] at (0,-2.3) {|};

\def\Rarrow{1.6}
\coordinate (mu1) at ({80:\Rarrow});
\coordinate (ml1) at ({280:\Rarrow});
\draw[RoyalBlue,-{Latex[length=2mm]}, line width=0.9pt]
  ($(mu1)-(0,0.2)$) -- (ml1)
  arc[start angle=280, end angle=80, radius=\Rarrow];
\coordinate (mu2) at ({54:\Rarrow});
\coordinate (ml2) at ({-54:\Rarrow});
\draw[RoyalBlue,-{Latex[length=2mm]}, line width=0.9pt]
  ($(mu2)-(0,0.2)$) -- (ml2)
  arc[start angle=-54, end angle=54, radius=\Rarrow];

\draw (v3) -- (c1) -- (c2) -- (c3) -- (v9);

\def\Rout{2.2}
\def\Rin{1.8}
\pgfmathsetmacro{\Rmid}{(\Rout+\Rin)/2}   
\pgfmathsetmacro{\Rcap}{(\Rout-\Rin)/2}   

\draw[dashed, line width=0.7pt]
    ({107:\Rout}) arc[start angle=107, end angle=253, radius=\Rout]
    ([shift={(253:\Rmid)}]253:\Rcap) arc[start angle=253, delta angle=180, radius=\Rcap]
    ({253:\Rin}) arc[start angle=253, end angle=107, radius=\Rin]
    ([shift={(107:\Rmid)}]-73:\Rcap) arc[start angle=-73, delta angle=180, radius=\Rcap];

\draw[dashed, line width=0.7pt]
    ({37:\Rout}) arc[start angle=37, end angle=-37, radius=\Rout]
    ([shift={(-37:\Rmid)}]-37:\Rcap) arc[start angle=-37, delta angle=-180, radius=\Rcap]
    ({-37:\Rin}) arc[start angle=-37, end angle=37, radius=\Rin]
    ([shift={(37:\Rmid)}]217:\Rcap) arc[start angle=217, delta angle=-180, radius=\Rcap];

\def\Rmidcap{0.2}
\draw[dashed, line width=0.7pt]
    ($(c1)+(-\Rmidcap,0)$) --
    ($(c3)+(-\Rmidcap,0)$)
    arc[start angle=-180, end angle=0, radius=\Rmidcap]
    -- ($(c1)+(\Rmidcap,0)$)
    arc[start angle=0, end angle=180, radius=\Rmidcap];
    
\draw[dashed, line width=0.7pt] (v9) circle[radius=0.2];
\draw[dashed, line width=0.7pt] (v3) circle[radius=0.2];

%% file: fig-labeling.tex
\def\n{10}
\def\r{2}

\foreach \i in {1,...,\n} {
\ifnum\i=3\else
    \node[circle, fill=black, inner sep=2pt] (v\i)
      at ({360/\n * (\i-1)}:\r) {};
\fi
}

\node[left, yshift=+0.2cm] at (v1) {\color{RoyalBlue} $\texttt{[}$};
\node[xshift=-0.3cm] at (v2) {\color{RoyalBlue} $\texttt{]}$};
\node[xshift=-0.3cm] at (v10) {\color{RoyalBlue} $\texttt{2}$};
\node[left, yshift=+0.4cm] at (v8) {\color{RoyalBlue} $\texttt{2}$};
\node[right, yshift=+0.2cm] at (v7) {\color{RoyalBlue} $\texttt{0}$};
\node[right, yshift=+0.2cm] at (v6) {\color{RoyalBlue} $\texttt{]}$};
\node[xshift=+0.4cm] at (v5) {\color{RoyalBlue} $\texttt{]}$};
\node[yshift=-0.4cm] at (v4) {\color{RoyalBlue} $\texttt{]}$};

\node[left, yshift=+0.3cm] at (v9) {\color{RoyalBlue} $\texttt{[}$};
\node[right, yshift=+0.3cm] at (v9) {\color{RoyalBlue} $\texttt{0}$};

\coordinate (v3) at ({360/\n * (3-1)}:\r);

\foreach \i in {1,...,\n} {
\pgfmathtruncatemacro{\j}{mod(\i,\n)+1}
\draw[bend left=-14] (v\i) to (v\j);
}

\path (v3) -- (v9)
node[pos=0.25, circle, fill=black, inner sep=2pt] (c1) {}
node[pos=0.50, circle, fill=black, inner sep=2pt] (c2) {}
node[pos=0.75, circle, fill=black, inner sep=2pt] (c3) {};

\node[left] at (c1) {\color{RoyalBlue} $\texttt{2}$};
\node[left] at (c2) {\color{RoyalBlue} $\texttt{[}$};
\node[left] at (c3) {\color{RoyalBlue} $\texttt{[}$};
\node[right] at (c1) {\color{RoyalBlue} $\texttt{0}$};
\node[right] at (c2) {\color{RoyalBlue} $\texttt{[}$};
\node[right] at (c3) {\color{RoyalBlue} $\texttt{]}$};

\draw (v3) -- (c1) -- (c2) -- (c3) -- (v9);

\draw[bend left=-20, line width=1pt] (v2) to (v10);
\draw[bend left=-23, line width=1pt] (v1) to (v10);
\draw[bend left=+35, line width=1pt] (c2) to (c3);
\draw[bend left=-25, line width=1pt] (c1) to (v4);
\draw[bend left=-35, line width=1pt] (c1) to (c2);
\draw[bend left=+50, line width=1pt] (v6) to (v8);
\draw[bend left=+10, line width=1pt] (v5) to (c3);
\draw[bend left=-20, line width=1pt] (v9) to (v8);

\draw[white] (0,-2.2)--(1,-2.2);

\node[right, yshift=0.2cm] at (v2) {\color{ForestGreen} $\texttt{[}$};
\node[xshift=0.3cm] at (v1) {\color{ForestGreen} $\texttt{]}$};
\node[right, yshift=-0.2cm] at (v10) {\color{ForestGreen} $\texttt{2}$};
\node[yshift=-0.35cm] at (v9) {\color{ForestGreen} $\texttt{[}$};
\node[xshift=-0.1cm, yshift=-0.3cm] at (v8) {\color{ForestGreen} $\texttt{2}$};
\node[xshift=-0.25cm, yshift=-0.25cm] at (v7) {\color{ForestGreen} $\texttt{0}$};
\node[left] at (v6) {\color{ForestGreen} $\texttt{]}$};
\node[xshift=-0.3cm, yshift=0.2cm] at (v5) {\color{ForestGreen} $\texttt{[}$};
\node[xshift=-0.1cm, yshift=0.35cm] at (v4) {\color{ForestGreen} $\texttt{]}$};

%% file: graph.tex
\begin{tikzpicture}
\small
\begin{axis}[
    axis lines = left,
    width = 0.8\textwidth,
    height = 0.4\textwidth,
    xlabel={$n$},
    ylabel={Upper bound},
    xmin=0, xmax=26,
    ymin=12.7, ymax=13.7,
    xtick={0,2,4,6,8,10,12,14,16,18,20,22,24},
    ytick={12.8, 13.0, 13.2, 13.4, 13.6},
    ymajorgrids=true,
    grid style=dashed,
    legend pos=north east,
    every axis plot/.append style={thick}
]

\addplot[
    color = teal,
    only marks,
    mark size=1.5pt,
] coordinates {
    (2, 13.5111143909180)
    (4, 13.3510455303084)
    (6, 13.1713030482988)
    (8, 13.0840250641339)
    (10, 13.0153319247146)
    (12, 12.9650694912555)
    (14, 12.9252724437665)
    (16, 12.8929131038170)
    (18, 12.8658715917401)
    (20, 12.8428189450607)
    (22, 12.8228533732631)
    (24, 12.8053315494523)
};

\addplot[black!60, thick, dashed, samples=1000, domain=1:26]
    {12.34 + 1.8/(x^0.426)};
    
\end{axis}
\end{tikzpicture}